\documentclass{llncs}

\usepackage[all]{xy}
\usepackage{graphicx}
\usepackage{cite}
\usepackage{amsmath}
\usepackage{amssymb}
\usepackage{amsfonts}
\usepackage{amscd}
\usepackage{stmaryrd}
\usepackage{hyperref}
\usepackage{version}
\usepackage{wrapfig}
\usepackage{subfig }
\usepackage{wasysym}
\usepackage{marvosym}

\newcommand{\actset}{\mathcal{A}}

\newcommand{\procset}{\mathsf{P}}

\newcommand{\exec}[1]{\mathrel{\xrightarrow{#1}}}

\newcommand{\nexec}[1]{\mathalpha{\not\exec{#1}}}
\newcommand{\pexec}[1]{\mathrel{\stackrel{#1}{\dashrightarrow}}}
\newcommand{\npexec}{\mathalpha{\not\dashrightarrow}}
\newcommand{\tick}{\mapsto}
\newcommand{\ntick}{\mathalpha{\not\mapsto}}
\newcommand{\success}{\omega}

\newcommand{\torp}[1]{#1}
\newcommand{\terma}[1]{\exec{a}\check}
\newcommand{\nterma}[1]{\nexec\check}
\newcommand{\termt}[1]{\tick\check}
\newcommand{\ntermt}[1]{\ntick\check}

\newcommand{\po}{ \parallel_L }

\def \lparal{\mathbin{\setbox0=\hbox{$\|$}
    \dimen0=\dp0 \advance\dimen0 -1.5pt \dp0=\dimen0
    \underline{\kern-1.5pt\box0\kern1.5pt}}}

\newcommand{\aproc}{{a}}

\newcommand{\deadlock}{\delta}

\newcommand{\sosaltact}{\mathalpha{\ensuremath{}}}

\newcommand{\sospchoiceone}{\mathalpha{\ensuremath{}}}
\newcommand{\sospchoicetwo}{\mathalpha{\ensuremath{}}}

\newcommand{\sospriority}{\mathalpha{\ensuremath{}}}
\newcommand{\sosrepact}{\mathalpha{\ensuremath{}}}
\newcommand{\sosparactone}{\mathalpha{\ensuremath{}}}
\newcommand{\sosparacttwo}{\mathalpha{\ensuremath{}}}
\newcommand{\sosparprobone}{\mathalpha{\ensuremath{}}}
\newcommand{\sosparprobtwo}{\mathalpha{\ensuremath{}}}

\newcommand{\sosencterm}{\mathalpha{\ensuremath{}}}
\newcommand{\sosencact}{\mathalpha{\ensuremath{}}}

\newcommand{\soshideterm}{\mathalpha{\ensuremath{}}}
\newcommand{\soshideactone}{\mathalpha{\ensuremath{\langle\textsf{hide-act}_1\rangle}}}
\newcommand{\soshideacttwo}{\mathalpha{\ensuremath{\langle\textsf{hide-act}_2 \rangle}}}

\newcounter{sosrule}
\newcommand{\sosrule}[2]{{\dfrac{#1}{#2}%
{\>\thesosrule \addtocounter{sosrule}{1}}}\ \ \ }

\newcommand{\sys}{\mathcal{P}}
\newcommand{\system}{(\nstates, \pstates, \atransym, \ptransym)}

\newcommand{\prob}{P}

  \newcommand{\bisimilar}[1][]{%
     \setbox0=\hbox{\kern-.1ex{$\leftrightarrow$}\kern-.1ex}
     \setbox1=\vbox{\hbox{\raise .1ex \box0}\hrule}%
     \ensuremath{\mathrel{\hbox{\kern.1ex\box1\kern.1ex}_{#1}}}
   }

\newcommand{\pbar}{\bar{p}}

\newcommand{\ptran}[1]{\pexec{#1}}

\newcommand{\ptransym}{\mathalpha{\dashrightarrow}}
\newcommand{\pitran}[1]{\ptran{#1}}

\newcommand{\atran}[1]{\exec{#1}}
\newcommand{\atransym}{\mathalpha{\rightarrow}}

\newcommand{\nstates}{S_n}
\newcommand{\pstates}{S_p}
\newcommand{\states}{S}

\newenvironment{LONGVERSION}{}{}
\newenvironment{SHORTVERSION}{}{}

\usepackage[all]{xy}
\xyoption{color} \xyoption{arc} \xyoption{matrix}\xyoption{frame}
\xyoption{dvips}
\newxyColor{mygray1}{0.7}{gray}{}
\newxyColor{mygray2}{0.9}{gray}{}
\newxyColor{mygray3}{0.5}{gray}{}
\newxyColor{mygray4}{0.0}{gray}{}
\newxyColor{mygray5}{0.3}{gray}{}
\newxyColor{mygray6}{0.7}{gray}{}
\newxyColor{mygray7}{0.45}{gray}{}

\newdir{...}{{}*{}*{}}
\newcommand{\slabl}[1]{\ar@(r,u)@{...}^{{#1\hspace{1.1mm}}}}
\newcommand{\slabr}[1]{\ar@(l,u)@{...}_{\hspace{1.1mm}#1}}

\newcommand{\shorten}[1]{}
\newcommand{\Result}{\mathsf{Res}}
\newcommand{\tests}{\mathcal{T}}
\newcommand{\menu}{M}
\newcommand{\obs}{\mathcal{O}}

\newcommand{\extch}{\mathsf{\sum}}
\newcommand{\prio}{\mathrel{\Theta}}
\newcommand{\barbed}{\mathrel{\approx_{\obs}}}
\newcommand{\testing}{\mathrel{\approx_{\mathcal{T}}}}
\newcommand{\algebra}{\textrm{CSP}_p}
\newcommand{\pst}{\circ}
\newcommand{\nst}{\circ}
\newcommand{\xmenu}{X}

\newcommand{\ratexpr}{\mathcal{R}}
\newcommand{\pop}{\bigoplus}

\newcommand{\poa}{\parallel}
\newcommand{\Det}{\mathsf{Det}}

\begin{document}

\title{Testing Probabilistic Processes: \\Can Random Choices Be Unobservable?}
\author{Sonja Georgievska and Suzana Andova}

\institute{{Department of Mathematics and Computer Science,
Eindhoven University of Technology, P.O.~Box
513, 5600 MB~~Eindhoven, The Netherlands} \\
{\tt \footnotesize s.georgievska@tue.nl}, {\tt \footnotesize
s.andova@tue.nl}}

\maketitle

\begin{abstract}
A central paradigm behind process semantics based on observability
and testing is that the exact moment of occurring of an internal
nondeterministic choice  is unobservable. It is natural, therefore,
for this property to hold when the internal choice is quantified
with probabilities. However, ever since probabilities have been
introduced in process semantics, it has been a challenge to preserve
the unobservability of the random choice, while not violating the
other laws of process theory and probability theory. This paper
addresses this problem. It proposes two semantics for processes
where the internal nondeterminism has been quantified with
probabilities. The first one is based on the notion of testing, i.e.
interaction between the process and its environment. The second one,
the probabilistic ready trace semantics, is based on the notion of
observability.  Both are shown to coincide. They are also preserved
under the standard operators.
\end{abstract}


\excludeversion{SHORTVERSION}

\section{Introduction}\label{sec:introduction}

A central paradigm behind process semantics based on observability
(e.g. \cite{csp}) is that the exact moment of occurring of an
internal nondeterministic choice is unobservable. This is because an
observer does not have insight into the internal structure of a
process but only in the externally visible actions. Unobservability
of  internal choice has been also accomplished by the testing theory
\cite{DH84}\footnote{In fact, the process semantics based on
\cite{csp} and \cite{DH84} do coincide for a broad class of
processes, as shown in \cite{D87}.}. It is natural, therefore, for
this property to hold when the internal choice is quantified with
probabilities. However, it turned out that unobservability of
internal probabilistic choice is not trivial to achieve in
probabilistic testing theory. To explain why, we start with an
example.

\paragraph{Motivation}Consider a machine which flips a fair coin
internally. A user can guess the result of the flipping by pressing
a ``head'' or a ``tail'' button.
 If the user has guessed correctly, the
machine offers a prize. The machine can be modeled by process graph
(or shortly process) $s$ in Fig. \ref{fig_machine}  and the user can
be modeled by process $u$ in Fig. \ref{fig_machine}. The user is
happy if, after pressing a button, a prize follows.

Let the user and the machine interact, i.e. let them synchronize on
all actions (except on the ``user happiness'' reporting action
$\smiley$). In terms of testing theory \cite{DH84},   process $s$ is
tested with test $u$.
 Intuitively, the probability that the user has guessed the
 output of flipping  is $\frac{1}{2}$. That is, the probability of a
 $\smiley$ action being reported is $\frac{1}{2}$. However,  most
  of the existing approaches for probabilistic testing, in particular probabilistic may/must testing
   \cite{YL92,JY02,segala_testing96,DGHM08, PNM07},
   do not give this answer.
Consider the synchronization $s \poa u$ represented by the graph in
Fig. \ref{fig_machine}, where actions are hidden after they have
synchronized.  In order to compute the probability of $\smiley$
being reported, the approaches in
\cite{YL92,JY02,segala_testing96,DGHM08, PNM07} use
\emph{schedulers}, that have insight into the internal structure of
the graph of the synchronized system. Each scheduler resolves the
nondeterminism in the nondeterministic nodes of $s \poa u$ and
yields a fully probabilistic system. For $s \mathrel{\poa} u$ in
Fig. \ref{fig_machine}, there are four possible schedulers, which
yield the following set of probabilities with which $s$ passes the
test $u$: $\{0, \frac{1}{2}, 1\}$. We can see that, because the
power of the schedulers is unrestricted, unrealistic upper and lower
bounds for the probabilities are obtained. Observe that this happens
due to the effect of ``cloning'' the nondeterminism after hiding the
synchronized actions. The choice between $h$ and $t$ has been
``cloned'' in both futures after the probabilistic choice in $s \poa
u$. When resolving nondeterminism in $s \poa u$, a scheduler assumes
that the user has unrealistic power to \emph{see} the result of the
coin-flipping \emph{before} guessing.

 The above example challenges us to reconsider the design choice to hide
actions after synchronization. Namely, although hiding is harmless
and actually useful in \cite{DH84}, and helps to abstract away from
unnecessary information, in probabilistic testing  it may actually
``hide too much'' and produce overestimation of the probability
information about the system. It is highly undesirable to obtain
lower and upper probability bounds of $0$ and $1$ resp. for the
probabilistic behaviour of a simple system (as the one in Fig.
\ref{fig_machine}), when the actual probability is $\frac{1}{2}$.
This may render a testing equivalence insufficient for verification
purposes.

\begin{figure}[t]\centering
\large
$\xymatrix@R=0.3cm@C=0.05cm{
    & & & s \\
    & & & \nst  \ar@{-->}[dll]_{\frac{1}{2}} \ar@{-->}[drr]^{\frac{1}{2}} \\
& \nst \ar[dl]_{h} \ar[dr]^{t} & & & & \nst  \ar[dl]_{h} \ar[dr]^{t} \\ 
          \nst \ar[d]^{p}  && \nst   & & \nst&& \nst\ar[d]_{p} \\ 
           \nst &&  & &  && \nst  \\ 
           }$
$\xymatrix@R=0.3cm@C=0.05cm{ u \\
   \nst \ar@/_/[d]_{h} \ar@/^/[d]^{t} \\ 
          \nst \ar[d]_{p} \\  
           \nst\ar[d]_{\smiley} \\ 
            \nst \\ 
           }$
           $\xymatrix@R=0.3cm@C=0.05cm{
&& s \poa u \\
     & & \nst  \ar@{-->}[dl]_{\frac{1}{2}} \ar@{-->}[drr]^{\frac{1}{2}} \\
& \nst \ar[dl]_{\tau} \ar[dr]^{\tau} & &  & \nst  \ar[dl]_{\tau} \ar[dr]^{\tau} \\ 
          \nst \ar[d]_{\tau}  && \nst   &  \nst&& \nst\ar[d]_{\tau} \\ 
           \nst\ar[d]_{\smiley} &&  & &  & \nst\ar[d]_{\smiley}  \\ 
           \nst &&  & &  & \nst  \\
           }$
         $\xymatrix@R=0.3cm@C=0.05cm{
         & & & \bar{s} \\
    & & & \nst  \ar[dll]_{h} \ar[drr]^{t} \\
& \pst \ar@{-->}[dl]_{\frac{1}{2}} \ar@{-->}[dr]^{\frac{1}{2}} & & & & \pst  \ar@{-->}[dl]_{\frac{1}{2}} \ar@{-->}[dr]^{\frac{1}{2}} \\ 
          \nst\ar[d]^{p}  && \nst  & & \nst&& \nst\ar[d]_{p} \\ 
          \nst &&  & &  && \nst   
           }$
 \caption{\small Processes $s$ and $\bar{s}$ are distinguished in  probabilistic may/must testing theory }
 \hbox{}\vspace{-1cm}
 \label{fig_machine}
\end{figure}

 Consider now  process
$\bar{s}$ in Fig. \ref{fig_machine}. To the user this graph may as
well represent the behaviour of the coin-flipping machine -- the
user cannot see whether the machine flips the coin \emph{before} or
\emph{after} making the ``head or tail'' guess. According to her,
the machine acts as specified as long as she is able to guess the
result in half of the cases. In fact, both schedulers applied to
$\bar{s} \poa u$  yield  that the probability of reporting a
$\smiley$ action is exactly $\frac{1}{2}$. Because of the last, none
of the approaches in \cite{YL92,JY02,segala_testing96,DGHM08,PNM07}
equate processes $s$ and $\bar{s}$, as, when tested with $u$, they
produce different bounds for the probabilities of reporting
$\smiley$.
 \footnote{If we ignore
the probabilities, processes $s$ and $\bar{s}$ are
testing-equivalent by \cite{DH84}.} Note that being able to equate
$s$ and $\bar{s}$ means allowing distribution of external choice
over internal probabilistic choice \cite{csp}.

Not allowing distribution of external choice over internal
probabilistic choice has an additional effect, undesirable for
compositional verification. Namely,
 if distribution of external choice over
internal probabilistic choice is not allowed, then distribution of
prefix over internal probabilistic choice  is questioned too, as
this implies congruence issues for \emph{asynchronous} or
\emph{concurrent} parallel composition \cite{csp} (where processes
synchronize on their common actions while  interleave on the other
actions). For instance, we would not be able to equate processes
$e.a. (b \oplus_{\frac{1}{2}}c)$ and $e.((a.b)
\oplus_{\frac{1}{2}}(a.c))$. (The operator ``.'' stands for
prefixing and the operator ``$\oplus$'' stands for a probabilistic
choice.) This  is because these two processes, running each
concurrently with process $e.d$, yield systems that cannot be
equated, unless we allow distribution of external choice over
internal probabilistic choice. If we are not able to relate
processes $e.a. (b \oplus_{\frac{1}{2}}c)$ and $e.((a.b)
\oplus_{\frac{1}{2}}(a.c))$, i.e. to allow distribution of prefix
over internal probabilistic choice, then for verification
 we can only rely on equivalences that inspect the internal
structure of processes, as bisimulations and simulations
\cite{spectrum1}, and, moreover, expect overestimation of
probabilities.

All together, the above discussions trigger the following question:
``In a model where the internal nondeterminism has been quantified
with probabilities \cite{LS91}, is it possible to test  process $s$
with
 test $u$ (Fig. \ref{fig_machine}) such that the result of
testing would imply that the probability of $s$ passing the test $u$
is exactly $\frac{1}{2}$?''. In this case not only  we could
preserve the information on  probability, but  we could also allow
distribution of  prefix  over  probabilistic choice without losing
compositionality.

\paragraph{Contributions} In this paper we show that the answer to
the above question is positive. The main contributions of the paper
are the following:
\begin{itemize}

\item We introduce a technique for labeling the synchronized actions
when a reactive probabilistic process is tested
(Section~\ref{sec:testing}). The labels are in form of rational
functions, whose argument names are constructed from the action
labels set. The labeling is achieved automatically when processes
synchronize, i.e. no additional manipulation on the process graphs
is needed.

\item We propose a testing semantics (Section~\ref{sec:testing}) exploiting the new labeling method,
such that the result of testing process $s$ with test $u$ in Fig.
\ref{fig_machine} is $\frac{1}{2}$, and processes $s$ and $\bar{s}$
in Fig. \ref{fig_machine} are testing-equivalent.

\item We define a probabilistic ready trace equivalence  for reactive
probabilistic processes using the Bayesian definition of probability
(Section~\ref{sec:barbed}). The definition allows a testing scenario
in the lines of \cite{spectrum1, CSV07} to be easily constructed.

\item We define an algebra of finite processes  and show that the ready trace equivalence is congruence for the standard operators (Section~\ref{sec:algebra}).

\item We show that  all operators of our algebra, including external choice, distribute over probabilistic choice, allowing us to consider the latter one as unobservable (Section~\ref{sec:algebra}).
\item We  show that the testing equivalence of Sec. \ref{sec:testing} and the ready-trace equivalence of Sec. \ref{sec:barbed} coincide
(Section~\ref{sec:relationship}).
\end{itemize}

Section \ref{sec:conclusion} ends with concluding remarks, future
work directions regarding coexistence of probabilistic and internal
choice, and related work.

 \begin{SHORTVERSION} An extended version
of this paper, containing the proofs of the main results, can be
found in the report \cite{testing_report}.\end{SHORTVERSION}

\section{Preliminaries}\label{sec:prelim}
We define some preliminary notions needed for the rest of the paper.
\paragraph{Bayesian probability} For a set $A$, $2^A$ denotes its power-set.
The following definitions are taken from~\cite{lindley}.

We consider a sample space, $\Omega$, consisting of points called
\emph{elementary events}. Selection of a particular $a \in \Omega$
is referred to as an ``$a$ has occurred''. An \emph{event} is a set
of elementary events. $A, B, C, \ldots$ range over events. An event
$A$ \emph{has occurred} iff  for some  $a \in A$ $a$ has occurred.
Let $A_1, A_2, \ldots$ be a sequence of events and $C$ be an event.
The members of the sequence are \emph{exclusive given C}, if
whenever $C$ has occurred no two of them can occur together, that
is, if $A_i \cap A_j \cap C=\emptyset$ whenever $i\not = j$. $C$ is
called a \emph{conditioning} event. If the conditioning event is
$\Omega$,
then ``given $\Omega$'' is  omitted.

 For certain pairs of events $A$ and $B$, a real number $P(A|B)$ is
 defined and called the \emph{probability} of $A$ given $B$. These
 numbers satisfy the following axioms:

 \begin{enumerate}
 \item[A1:] $0 \leq P(A|B)\leq1$ and $P(A|A)=1$.
 \item[A2:]  If the events in $\{A_i\}_{i=1}^{\infty}$ are
 exclusive given $B$, then  $P(\cup_{i=1}^{\infty}
 A_i\ |\ B)=\sum_{i=1}^{\infty}P(A_i | B).$
 \item[A3:]  $P(C|A \cap B)\cdot P(A|B)=P(A \cap C|B)$.
 \end{enumerate}
For $P(A|\Omega)$ we simply write $P(A)$.

 \paragraph{Probabilistic transition systems} In a probabilistic transition system (PTS)  there are two types of  transitions,
 viz. action and probabilistic transitions;
 a state can either perform action transitions only
(nondeterministic state) or (unobservable) probabilistic transitions
only (probabilistic state). To simplify, we assume that
probabilistic transitions lead to nondeterministic states. The
nondeterministic states exhibit only a so-called \emph{external}
(observable) nondeterminism, i.e the choice is between the actions,
but once the action is chosen, the next  state is determined. The
outgoing transitions of a probabilistic state $s$ define probability
over the power-set of the set of nondeterministic states.

We give a formal definition of a PTS. Presuppose a finite set of
actions $\actset$.

\begin{definition}[Probabilistic Transition System (PTS)]
\label{def:PTS} A \emph{PTS} is a tuple $\sys=\system$, where
\begin{itemize}
\item $\nstates$ and $\pstates$ 
are finite disjoint sets of \emph{nondeterministic} and
\emph{probabilistic states}, resp.,
\item $\atransym\subseteq
\nstates \times \actset \times \nstates \cup \pstates$ is an
\emph{action transition relation} such that $(s,a,t) \in \atransym$
and $(s,a,t') \in \atransym$ implies $t=t'$, and
\item $\ptransym \subseteq
\pstates \times (0,1] \times \nstates$ is a \emph{probabilistic
transition relation} such that, for all $s \in \pstates$,
$\sum_{(s,\pi,t)\in \ptransym}\pi =1$.

\end{itemize}
\end{definition}

\noindent We denote $\nstates \cup \pstates$ by $\states$. We write
$s \atran{a} t$ rather than $(s,a,t) \in \atransym$, and $s
\ptran{\pi} t$  rather than $(s,\pi,t)\in \ptransym$ (or $s \ptran{}
t$ if the value of $\pi$ is irrelevant in the context).  We write $s
\atran{a}$ to denote that there exists an action transition $s
\atran{a}s'$ for some $s' \in \states$. We agree that a state
without outgoing transitions
belongs to $\nstates$.


As standard, we define a \emph{process graph} (or simply
\emph{process}) to be a state $s \in \states$ together with all
states reachable from $s$, and the transitions between them. A
process graph is usually named by its \emph{root} state, in this
case $s$. 


\section{Testing equivalence} \label{sec:testing}

In this section we define a testing equivalence in the style of
\cite{DH84} for reactive probabilistic
processes. 

Recall from elementary mathematics that a division of two
polynomials is called a \emph{rational function}. For example,
$\frac{2x}{x+y}$ is a rational function with  arguments $x$ and $y$.
A possible domain for this function is $(0, \infty) \times (0,
\infty)$. We are going to exploit a subset $\ratexpr$ of the
rational functions
 whose argument names
belong to the action labels $\actset$, which is generated by the
following grammar:

\[\varphi::=\alpha  \ \mid \   a \ \mid \ \varphi + \varphi \ \mid \
\varphi \cdot \varphi \ \mid \ \frac{\varphi}{\varphi},\]

\noindent where $\alpha$ is a non-negative scalar, $a \in \actset$,
and $+, \ \cdot, \ $ and $ \frac{\cdot}{\cdot}$ are ordinary
algebraic addition, multiplication and fraction, resp. Brackets are
used in the standard way to change the priority of the operators.
For our purposes, we assume that the arguments  $a,b, ...$ can only
take positive values, i.e. the domain of every function in
$\ratexpr$ is $(0, \infty)^{n}$, where $n$ is the size of the action
set. Therefore, two rational functions in $\ratexpr$ are equal iff
they can be transformed to equal terms using the standard
transformations that preserve equivalence (e.g. for $a,b \in
\actset$, $ \frac{1}{2}\cdot \frac{a}{a+b}+\frac{1}{2}\cdot
\frac{b}{a+b}=\frac{1 \cdot (a+b)}{2 \cdot(a+b)}=\frac{1}{2}$).

As standard, a \emph{test} $T$ is a finite process
 such that,
 for a symbol $\success \not \in \actset$, there may exist
transitions $s \atran{\success} $
 for some  states $s$ of $T$. Denote the set of all tests by $\tests$. Given a process $s$ and
action $a \in \actset$, denote by $s_a$ the process (if exists) for
which $s \atran{a}s_a$. Given a PTS $\sys = \system$, let $I\colon
\nstates \mapsto 2^\actset$ be a function such that, for all $a \in
\actset, s \in \nstates$, it holds $a \in I(s)$ iff $s \atran{a}$.
$I (s)$ is called the \emph{menu} of $s$. Intuitively, for $s \in
\nstates$, $I(s)$ is the set of actions that the process  $s$ can
perform initially. Next, we define the result of testing a process
with a given test. The informal explanation follows afterwards.

\begin{definition}\label{def:Results} The function
$\Result \colon \states \times \tests \mapsto \ratexpr$ that gives
the result of testing a process $s$ with a test $T$ is defined as
follows:

\[
\Result(s,T)=
\begin{cases} {}
1, & \text{if } T \atran{\success}, \\

\sum_{i \in I}{\pi_i \cdot \Result(s_i , T)}, &  \text{if } s
\pitran{\pi_i}s_i \text{ for } i \in I$ \text{and } $T  \not \atran{\success} \\

\sum_{i \in I}{\pi_i \cdot \Result(s , T_i)}, &  \text{if } T
\pitran{\pi_i}T_i \text{ for } i \in I$ \text{and } $s  \not \ptran{} \\


\sum_{a \in K} \frac{a}{\sum_{b \in K}{b}} \cdot \Result(s_a , T_a),
& \text{for } K = I(s) \cap I(T),\text{ otherwise}.

\end{cases}
\]
\end{definition}

\noindent As usual, the result of testing a process  with a test
denoting success is one, while the result of testing a process with
a probabilistic state as a root (i.e. initially probabilistic
process) is a weighted sum of the results of testing the subsequent
processes with the same test. Similarly when the test is initially
probabilistic. The novelty is in the result of testing an initially
nondeterministic process $s$ with a test $T$ that can initially
perform actions from $\actset$ only. Namely, when the process and
the test synchronize on an action, the resulting transition is
labeled with a ``weighting factor'', containing information about
the way this synchronization happened.
 \begin{wrapfigure}{r}{39mm}
\centering \large \hbox{}\vspace{-1cm}
           $\xymatrix@R=0.3cm@C=0.07cm{
    & & & \nst  \ar@{-->}[dll]_{\frac{1}{2}} \ar@{-->}[drr]^{\frac{1}{2}} \\
& \nst \ar[dl]_{\frac{h}{h+t}} \ar[dr]^{\frac{t}{h+t}} & & & & \nst  \ar[dl]_{\frac{h}{h+t}} \ar[dr]^{\frac{t}{h+t}} \\ 
          \nst \ar[d]_{\frac{p}{p}}  && \nst   & & \nst&& \nst\ar[d]_{\frac{p}{p}} \\ 
           \nst\ar[d]_{\smiley} &&  & &  && \nst\ar[d]_{\smiley}  \\ 
           \nst &&  & &  && \nst  \\
           }$
 \caption{\small Graphical representation of the  result of testing $s$ (Fig. \ref{fig_machine})  with $u$}
\hbox{}\vspace{-1cm}
 \label{fig_testing}
\end{wrapfigure}
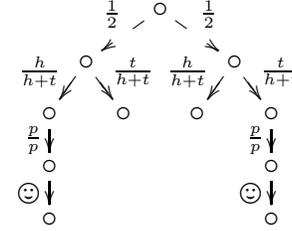
This information has form of a rational function, the numerator of
which represents the synchronized action itself, while the
denominator is the sum of the common initial actions of $s$ and $T$,
i.e., all actions on which $s$ and $T$ could have synchronized at
the current step. Then, the rational function is temporarily treated
as ``symbolic'' probability, in order to compute the final result of
the testing. The final result is again a rational function in
$\ratexpr$.

Fig. \ref{fig_testing} represents graphically the result of testing
process $s$ in Fig. \ref{fig_machine} with the test $u$ from the
same figure. It is easy to compute that the result of testing is
equal to $\frac{1}{2}$, which establishes one of our goals set in
Section \ref{sec:introduction}. However, in many cases the result is
a non-scalar rational function. For example, denote by ``$+$'' the
external choice operator. The result of applying  test $h.p.\success
+t.\success$  to each of processes $s$ and $\bar{s}$ in Fig.
\ref{fig_machine} is $\frac{h+2t}{2(h+t)}$.

\begin{definition} \label{def:testing} Two processes $s$ and $\bar{s}$ are \emph{testing
equivalent}, notation $s \testing \bar{s}$, iff $\Result(s,T)$ and
$\Result(\bar{s},T)$ are equal functions for every test $T$.
\end{definition}

\noindent Obviously, comparing two results boils down to comparing
two polynomials, after both rational functions have been transformed
to equal denominators.

\begin{example}Consider the processes in Fig. \ref{fig_notequiv}. The test
$a.\success + b.c.\success$ distinguishes between the two processes.
\end{example}
\begin{figure}[t]\centering
\large \hbox{}\vspace{-0.7cm}
     $
    \xymatrix@R=0.3cm@C=0.05cm{
    & & & \pst  \ar@{-->}[dll]_{\frac{1}{2}} \ar@{-->}[drr]^{\frac{1}{2}} \\
& \nst \ar[dl]_{a} \ar[dr]^{b} & & & & \nst  \ar[dl]_{b} \ar[dr]^{e} \\ 
          \nst  && \nst\ar[d]_{c}  & & \nst\ar[d]_{d} && \nst \\ 
           && \nst & & \nst &&   \\ 
           }
           \quad \not \testing \quad
           \xymatrix@R=0.3cm@C=0.05cm{
    & & & \pst  \ar@{-->}[dll]_{\frac{1}{2}} \ar@{-->}[drr]^{\frac{1}{2}} \\
& \nst \ar[dl]_{a} \ar[dr]^{b} & & & & \nst  \ar[dl]_{b} \ar[dr]^{e} \\ 
          \nst  && \nst\ar[d]_{d}  & & \nst\ar[d]_{c} && \nst \\ 
           && \nst & & \nst &&   \\ 
           }$
\newline
\hbox{}\vspace{-0.3cm}
 \caption{Processes $s$ (left) and $\bar{s}$ (right) are not testing equivalent}
 \hbox{}\vspace{-0.7cm}
 \label{fig_notequiv}
\end{figure}
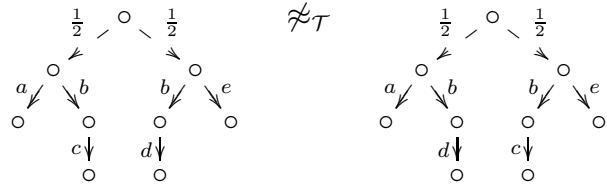
\begin{remark} Def. \ref{def:Results} assumes that, when the process and the
test are ready to synchronize on an action,  the  test can
\emph{see} which actions have been offered from the process. This
corresponds to the user (e.g. $u$ in Fig. \ref{fig_machine}) being
able to see the menu that the machine (e.g. $s$ in Fig.
\ref{fig_machine}) offers. Note that this assumption does not exist
in the standard non-probabilistic testing theory \cite{DH84}.
However, in real-life systems this is usually the case. Moreover,
this assumption is mild with respect to probabilistic may/must
testing approaches, where one needs to know the complete internal
structure of the composed process, which, on the other side, yields
unrealistic over-estimations of probabilities. In contrary, in our
case, in order to compute the function $\Result(s,T)$, it is not
necessary that the probabilistic transitions of $s$ and their labels
are known. Their effect can be inferred
 statistically, by testing $s$ with $T$ sufficiently many times. To
 simplify the presentation, we do not go into details on statistical testing.
\end{remark}

\section{Probabilistic ready trace semantics}\label{sec:barbed}

In this section we define a probabilistic version of ready trace
equivalence~\cite{Pnueli85, BBKready}.

\begin{definition}[Ready trace] A \emph{ready trace of length $n$} is a  sequence
 $\obs = (
 \menu_1, a_1, \menu_2, a_2, \ldots, \menu_{n-1}, a_{n-1}, \menu_n)$ where $\menu_i \in 2^\actset$ for all
$i \in \{1, 2, \ldots , n\}$ and $a_i \in \menu_i$ for all $i \in
\{1,2, \ldots, n-1\}$ . 
\end{definition}

\noindent We assume that  the observer has  ability to observe the
actions that the process performs, together with the menus out of
which actions are chosen. Intuitively, a ready trace $\obs = (
 \menu_1, a_1, \menu_2, a_2, \ldots, \menu_{n-1}, a_{n-1}, \menu_n)$ can be
observed if the initial menu is $\menu_1$,  then action $a_1 \in
\menu_1$ is performed, then the next menu is $\menu_2$,  then action
$a_2 \in \menu_2$ is performed and so on, until the observing ends
at a point when the  menu is $\menu_n$. It is essential that, since
the probabilistic transitions are not observable, the observer
cannot infer where exactly they happen in the ready trace.

Clearly the probability of observing a ready trace
$(\{a,b\},a,\{c\})$ is conditioned on choosing the action $a$ from
the menu $\{a,b\}$.
 This suggests that, when
defining probabilities on ready traces, the Bayesian definition of
probability is more appropriate than the measure-theoretic
definition that is usually taken.

Next, given a  process $s$, we define a process $s_{(\menu, a)}$.
Intuitively, $s_{(\menu, a)}$ is  the process that $s$ becomes,
assuming that menu $\menu$ was offered to $s$ and action $a$ was
performed.

\begin{definition}\label{def:pafter}
Let $s$ be a process graph. Let $\menu \subseteq \actset$, $a \in
\menu$ be such that $I(s)=\menu$ if $s \in \nstates$ or otherwise
there exists a transition $s \ptran{}s'$ such that $I(s')=\menu$.
The process graph $s_{(\menu,a)}$ is obtained from $s$ in the
following way:
\begin{itemize} \item if $s \in \nstates$ then the root of
$s_{(\menu,a)}$ is the state $s'$ such that $s \atran{a} s'$, and
\item if $s \in \pstates$ then a new state $s_{(\menu, a)}$ is
created. Let $\pi = $
$\sum_{s\ptran{\pi_i}s_i,I(s_i)=\menu}{\pi_i}$. For all $s_i'$ such
that $s \ptran{\pi_i}s_i\atran{a}s_i'$ and $I(s_i)=\menu$:

\begin{itemize} \item if $s_i' \not \ptran{}$, then an edge
$s_{(\menu,a)}\ptran{\pi_i/\pi}s_i'$  is created;
\item for all transitions $s_i' \ptran{\rho_i}s_i''$, an edge
$s_{(\menu,a)}\ptran{\pi_i \rho_i/\pi}s_i''$  is created.
\end{itemize}
\end{itemize}
\end{definition}

\begin{example} Consider processes  $s$ and $s_{(\{a,b\},a)}$ in Fig. \ref{fig_pafter}.
Assuming that the initial menu of $s$ was $\{a,b\}$ and action $a$
was performed,   process  $s_{(\{a,b\},a)}$  describes the further
behaviour of $s$: with  probability
$\frac{1}{8}/(\frac{1}{8}+\frac{3}{8}) = \frac{1}{4}$ action
 $c$ is performed, while with  probability
$\frac{3}{8}/(\frac{1}{8}+\frac{3}{8}) = \frac{3}{4}$
action  $d$ is performed.
\end{example}

\begin{figure}[t]\centering
\large
       $ \xymatrix@R=0.3cm@C=0.1cm{
          & & & & & & \pst \ar@{-->}[dlllll]_{\frac{1}{8}} \ar@{-->}[dll]^{\frac{3}{8}} \ar@{-->}[drr]_{\frac{1}{4}} \ar@{-->}[drrrrr]^{\frac{1}{4}} & & & & & &\\
          & \nst \ar[dl]_{a} \ar[dr]^{b} & & & \nst \ar[dl]_{a} \ar[dr]^{b} & & & & \nst \ar[dl]_{a} \ar[dr]^{c} &&&  \nst \ar[dl]_{a} \ar[dr]^{e} & \\
          \nst\ar[d]_{c}  && \nst\ar[d]_{e} & \nst\ar[d]_{d} && \nst\ar[d]_{f} && \nst\ar[d]_{d} && \nst\ar[d]_{e} & \nst\ar[d]_{d} && \nst\ar[d]_{f} \\
          \nst && \nst & \nst && \nst && \nst &&  \nst & \nst && \nst
        }$
        \quad \quad
$\xymatrix@R=0.3cm@C=0.1cm{
& \nst \ar@{-->}[dl]_{\frac{1}{4}} \ar@{-->}[dr]^{\frac{3}{4}} \\ 
          \nst\ar[d]_{c}  && \nst\ar[d]_{d}  \\
          \nst  && \nst
           }
           $
\newline
 \caption{Example of a process $s$ (left) and $s_{(\{a,b\},a)}$
 (right).}
 \hbox{}\vspace{-0.5cm}
 \label{fig_pafter}
\end{figure}
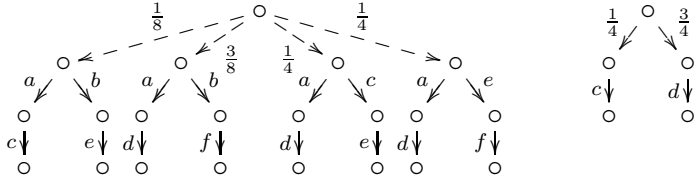


\begin{definition} \label{def:P}Let $(
 \menu_1, a_1, \menu_2, a_2, \ldots, \menu_{n-1}, a_{n-1}, \menu_n)$
 be a ready trace of length $n$ and $s$ be a process graph.
 Functions $P_s^1(\menu )$  and
 $\prob_s^n(\menu_n | \menu_1, a_1, \ldots \menu_{n-1}, a_{n-1})$ (for $n > 1$)
 are defined in the following way:\\

\noindent$ \prob_s^1(\menu )=
\begin{cases}{}
\sum_{s \ptran{\pi}s'}{\pi \cdot P_{s'}^1(
\menu )} & \textrm{if } s \in \pstates,\\
1 & \textrm{if } s \in \nstates, \ I(s)=\menu, \\
0 & \textrm{otherwise. } \\
\end{cases}$\\\\
\noindent$\prob_s^2(\menu_2 | \menu_1, a_1) =
\begin{cases}{}
\prob_{s_{(\menu_1,a_1)}}^1(\menu_2) & \textrm{if} \ \prob_s^1(\menu_1)>0, \\
\textrm{undefined} & {otherwise}.
\end{cases}
$\\\\
\noindent$\prob_s^n(\menu_n | \menu_1, a_1, \ldots, a_{n-1})=
\begin{cases}{}
\prob_{s_{(\menu_1,a_1)}}^{n-1}(\menu_n | \menu_2, a_2, \ldots,
 a_{n-1}) & \textrm{if} \ \prob_s^1(\menu_1)>0, \\
\textrm{undefined} & {otherwise}.
\end{cases}
$
\end{definition}

Let the sample space consist of all possible menus and $s \in
\states$. Function $P_s^1(\menu)$ can be interpreted as the
probability that the menu $\menu$ is observed initially when process
$s$ starts executing.  Let the sample space consist of all ready
traces of length $n$ and let $s \in \states$. The function
$P_s^n(\menu_n | \menu_1, a_1, \ldots \menu_{n-1}, a_{n-1})$ can be
interpreted as the probability of the event $ \{(\menu_1, a_1,
\ldots, \menu_{n-1}, a_{n-1}, \menu_n) \} $, given the event $
\{(\menu_1, a_1, \ldots \menu_{n-1}, a_{n-1}, \xmenu) \textrm{~:~}
\xmenu \in 2^{\actset} \}$, if observing ready traces of process
$s$. It can be checked that these probabilities are well defined,
i.e., they satisfy the axioms A1-A3 of Section \ref{sec:prelim}.

\begin{definition}[Probabilistic ready trace
equivalence]\label{def:barbed} Two processes $s$ and $\bar{s}$ are
\emph{probabilistically ready trace equivalent}, notation $s \barbed
\bar{s}$, iff:
\begin{itemize}
\item for all $\menu$ in $2^\actset$, $P_s^1(\menu)=P_{\bar{s}}^1(\menu)$ and
\item for all $n > 1$, $P_s^n(\menu_n | \menu_1, a_1, \ldots \menu_{n-1},
a_{n-1})$ is defined if and only if $P_{\bar{s}}^n(\menu_n |
\menu_1, a_1, \ldots \menu_{n-1}, a_{n-1})$ is defined, and in  that
case $P_s^n(\menu_n | \menu_1, a_1, \ldots \menu_{n-1}, a_{n-1})$
$=$ $ P_{\bar{s}}^n(\menu_n | \menu_1, a_1, \ldots \menu_{n-1},
a_{n-1})$.
\end{itemize}
\end{definition}

\noindent Informally, two processes $s$ and $\bar{s}$ are
ready-trace equivalent iff for every $n$ and every ready trace
$(\menu_1, a_1, \menu_2, a_2 , \ldots \menu_n )$ , the probabilitity
to observe $\menu_n$, under condition that previously the sequence
$(\menu_1, a_1, \menu_2, a_2 , \ldots a_{n-1} )$ was observed, is
defined at the same time for both $s$ and $\bar{s}$; moreover, in
case both probabilities are defined, they coincide. Note that it is
straightforward to construct a black-box testing scenario
\cite{spectrum1, CSV07} for this ready-trace equivalence.

\begin{example} Processes $s$ and $\bar{s}$ in Fig. \ref{fig_machine} are  ready-trace equivalent. Processes in Fig. \ref{fig_notequiv} are not
ready-trace equivalent: for process $s$ it holds $\prob_s^2(\{c\} |
\{a,b\}, b)=\frac{1}{2}$, while for  process $\bar{s}$ it holds
$\prob_t^2(\{c\} | \{a,b\}, b)=0$.
\end{example}

\section{Algebra}\label{sec:algebra}


In this section we define an algebra $\algebra$ of finite processes
using $\barbed$ as an underlying equivalence. The purpose is to show
that $\barbed$ is congruence for the standard operators on the model
of reactive probabilistic processes and that all operators
distribute through probabilistic choice, as all operators distribute
through  internal choice in standard CSP~\cite{csp}. As discussed in
Sec. \ref{sec:introduction}, we do not use hiding operator. For more
discussions on including internal nondeterminism in general, please
see Sec. \ref{sec:conclusion}.

The set of $\algebra$ processes $\procset$ is generated by the
following grammar:
\[\begin{array}{lcl}
 \procset ::=\ \deadlock \ \mid\
          \extch_{i \in I}{a_i}.\procset_i \mid\
          \pop_{i \in I}{\pi_i}\procset_i \mid\
          \prio \procset \mid\
          \procset\poa \procset \mid\
          \procset\po \procset

\enskip
\end{array}\] where $\deadlock \not \in  \actset$ is a new symbol,   $\{a_i\}_{i \in I} \subseteq \actset$,
$a_i \not = a_j$ for $i,j \in I, \ i \not = j$,  $\pi_i \in (0,1]$,
$\sum_{i \in I}{\pi_i}=1$, and $L \subseteq \actset$ is the set of
actions that appear both in the left and in the right process  of
the expression
$\procset\po \procset$. 

Let $p,q,r,...$  range over $\algebra$ processes. The constant
$\deadlock$ stands for the \emph{empty} process. The process
$\aproc.p$ performs the action $a$ and continues as process $p$ (we
write $a$ rather than $a.\deadlock$). The \emph{external choice}
$\extch_{i \in I}{a_i}.p_i$ stands for a choice among the actions
$\{a_i\}_{i \in I}$ and proceeds as process $p_j$ if action $a_j$ is
chosen and executed. The \emph{probabilistic choice} $\pop_{i \in
I}{\pi_i}p_i$ behaves as $p_i$ with probability $\pi_i$ for $i \in
I$. The \emph{priority} operator $\prio$  assumes a partial order
$>$ on $\actset$. For actions $a$ and $b$,
 we say $a$ has higher priority than $b$ iff $a > b$.  $\prio$ forces the process  to always perform the action with the highest
priority in the current menu. In a \emph{synchronized parallel
composition} $p \poa q$, the processes operate in a lock-step
synchronization. In a \emph{parallel composition} $p \po q$, the
processes synchronize on their common actions, while the other
actions are interleaved. \footnote{To preserve associativity of
$\po$, we  require that for any processes $p,q$, and $r$, if $p$ and
$q$ share actions and $q$ and $r$ share actions then $p$ and $r$ do
not share actions.} \footnote{Sequential composition and successful
termination can be also defined, which we avoid here to shorten.}

Table~\ref{table:procsem} represents the operational semantics of
$\algebra$ processes (we omit the symmetric rules for $\po$ and
$\poa$).

\begin{table}[t!]
 \center\framebox{{\small
\begin{tabular}{c}


\renewcommand{\thesosrule}{\sosaltact}
    $\sosrule{}{\extch_{i \in I}{a_i}.p_i\exec{a_i}p_i}$

       $\renewcommand{\thesosrule}{\sosencact}
    \sosrule{p\exec{a} p',\ q\exec{a}{q'}}
            {p\poa q\exec{a} p'\poa q'}$

  $ \renewcommand{\thesosrule}{\soshideterm}
    \sosrule{p \pexec{\pi} p',\ q \pexec{\rho} q'}
            {p\poa q \pexec{\pi \rho} p' \poa q'}$


  \\\\

  $ \renewcommand{\thesosrule}{\sosencterm}
    \sosrule{p \pexec{\pi} p',\ q \npexec}
            {p\poa q \pexec{\pi} p' \poa q}$
%
%
%


    $\renewcommand{\thesosrule}{\sospchoiceone}
    \sosrule{ p_k\npexec, k \in I}
            {\pop_{i \in I}{\pi_i p_i}\ptran{\pi_k}p_k}$

            $\renewcommand{\thesosrule}{\sospchoicetwo}
    \sosrule{ p_k\pexec{\rho_k}p_k', k \in I}
            {\pop_{i \in I}{\pi_i p_i}\ptran{\pi_k\rho_k}p_k'}$
\\\\

    $\renewcommand{\thesosrule}{\sosparactone}
    \sosrule{a \not \in L, \ p\exec{a}\torp{p'},\ q\npexec}{p\po q\exec{a} {p'\po q}}$

       $\renewcommand{\thesosrule}{\sosparacttwo}
    \sosrule{a \in L, \ p\exec{a} p',\ q\exec{a}{q'}}
            {p\po q\exec{a} p'\po q'}$

$ \renewcommand{\thesosrule}{\sosparprobone}
    \sosrule{p \pexec{\pi} p',\ q \npexec}
            {p\po q \pexec{\pi} p' \po q}$

\\\\
  $ \renewcommand{\thesosrule}{\sosparprobtwo}
    \sosrule{p \pexec{\pi} p',\ q \pexec{\rho} q'}
            {p\po q \pexec{\pi \rho} p' \po q'}$

$ \renewcommand{\thesosrule}{\sospriority}
    \sosrule{p \exec{a} p',\ p \not \exec{b} \textrm{ for } a < b}
            {\prio p \exec{a} \ \prio p'}$

$ \renewcommand{\thesosrule}{\sosrepact}
    \sosrule{p \ptran{\pi} p'}
            {\prio p \ptran{\pi} \ \prio p'}$\\\\

%
%

\end{tabular} } } \caption{Operational semantics for $\algebra$ processes}\label{table:procsem}
\end{table}
As usual, a context is a $\algebra$ process  with a hole in it.
Given a context $C[\cdot]$ and a process $p$, we write $C[p]$ to
denote the process obtained by filling in the hole of $C[\cdot]$
with $p$. \begin{SHORTVERSION}The proofs of the following theorems
are given in \cite{testing_report}.\end{SHORTVERSION}

\begin{theorem}[Congruence]\label{thm:congruence} The equivalence $\barbed$ is congruence for
the operators of $\algebra$, i.e., if $p \barbed \pbar$ then for
each context $C[\cdot]$, it holds that $C[p] \barbed C[\pbar]$.
\end{theorem}
\begin{LONGVERSION}
\begin{proof}
We prove the congruence result for parallel composition, because
this is the most complicated case. We prove that if $p \barbed
\pbar$ then $p \po q \barbed \pbar \po q$.  Denote by $L$ the set of
the common actions for $p$ and $q$ (and therefore $\pbar$ and $q$).
Without loss of generality, assume that $p,\pbar, $ and $q$ are
probabilistic processes. For arbitrary menus $\menu', \menu''$,
denote by $\menu' \otimes \menu''$ the menu $(\menu' \cup
\menu'')\setminus
(L \setminus (\menu' \cap \menu''))$.\\

By induction on $n$, we prove that if $p \barbed \pbar$ then
$\prob_{(p \po  q)}^{n} (\menu_n | \menu_1, a_1, \ldots
\menu_{n-1},a_{n-1})$ $=$ $\prob_{(\pbar \po q)}^{n} (\menu_n |
\menu_1, a_1, \ldots \menu_{n-1},a_{n-1})$.

For arbitrary menus $\menu_p$ and  $\menu_q$, we have
$\prob_p^1(\menu_p)=\prob_{\pbar}^1(\menu_p)$. Let $\menu$ be a menu
such that $\prob_{p \po q}^1(\menu)>0$. This means that there exist
menus $\menu_p, \menu_q$ such that $\prob_p^1(\menu_p)>0$,
$\prob_q^1(\menu_q)>0$, and $\menu=\menu_p \otimes \menu_q$ (by
Table \ref{table:procsem}). We have,
\begin{align*}
&\prob_{p \po q}^1(\menu)=  \sum_{ \begin{scriptsize}
\begin{array}{c}
                                     p \po q \ptran{\lambda_k}r_k, \\
                                    I(r_k)=\menu
                                   \end{array}\end{scriptsize}
}\lambda_k =\sum_{ \begin{scriptsize}\begin{array}{c}
                      p \ptran{\pi_i} p_i,q \ptran{\rho_j}q_j, \\
                     I(p_i)\otimes I(q_j)=\menu
                   \end{array}\end{scriptsize}
 }{\pi_i \cdot \rho_j}\\
=&\sum_{q\ptran{\rho_j}q_j}{\rho_j\sum_{
p\ptran{\pi_i}p_i,\menu=I(p_i)\otimes I(q_j)}{\pi_i}}
=\sum_{q\ptran{\rho_j}q_j}{\rho_j\sum_{\pbar\ptran{\bar{\pi}_i}\bar{p}_i,\menu=I(\bar{p}_i)\otimes
I(q_j)}{\bar{\pi}_i}} = \prob_{\pbar \po q}^1(\menu).
\end{align*}
Suppose $\prob_{(p \po  q)}^{k} (\menu_k | \menu_1, a_1, \ldots
\menu_{k-1},a_{k-1})$$=$$\prob_{(\pbar \po q)}^{k} (\menu_k |
\menu_1, a_1, \ldots \menu_{k-1},a_{k-1})$ if $p \barbed \pbar$ and
$k<n$.

\noindent \textbf{Case 1} Suppose first that both $\prob_{(p \po
q)}^{n} (\menu_n | \menu_1, a_1, \ldots \menu_{n-1},a_{n-1})$ and
$\prob_{(\pbar \po q)}^{n} (\menu_n | \menu_1, a_1, \ldots
\menu_{n-1},a_{n-1})$ are defined.
 Because of Def. \ref{def:P} and the inductive assumption, it is enough to prove that $\prob_{(p \po
q)_{(\menu_1,a_1)}}^{n-1} (\menu_n | \menu_2, a_2, \ldots
\menu_{n-1},a_{n-1})$ = $\prob_{(\pbar \po q)_{(\menu_1,a_1)}}^{n-1}
(\menu_n | \menu_2, a_2, \ldots \menu_{n-1},a_{n-1})$. Because of
the inductive assumption, to prove the last, it is enough to prove
that  $(p
\po q)_{(\menu_1,a_1)} \barbed (\pbar \po q)_{(\menu_1,a_1)}$ .\\

\noindent \textbf{Case 1.1} \ $a_1=a \in L$.\\\\

\noindent Denote $\sum_{p \ptran{\pi_i}p_i,q\ptran{\rho_j}q_j,
I(p_i)\otimes I(q_j) = \menu_1} \pi_i \rho_j$ by $\alpha$. By Def.
\ref{def:pafter} and the rules in Table \ref{table:procsem}, we have
\begin{equation} \label{eq:cong1}
(p \po q)_{(\menu_1,a)} \equiv
 \pop_{\begin{scriptsize}\begin{array}{c}
          p \ptran{\pi_i}p_i,q\ptran{\rho_j}q_j, \\
         I(p_i)\otimes I(q_j) = \menu_1
       \end{array}\end{scriptsize}
}{\frac{\pi_i\rho_j}{\alpha} \Bigl(p_{i_{(I(p_i),a)}}\po
q_{j_{(I(q_j),a)}}}\Bigr).
\end{equation}
\noindent On the other hand, denoting $\sum_{\menu_1=\menu_p\otimes
\menu_q}{\prob_p(\menu_p) \prob_q(\menu_q)}$ by $\beta$, we have
\begin{align}\label{eq:cong2}
 &\  \pop_{
\menu_1=\menu_p\otimes \menu_q}
\frac{\prob_p(\menu_p)\prob_q(\menu_q)}{\beta
}{\Bigl(p_{(\menu_p,a)}\po q_{(\menu_q,a)}\Bigr)} \notag \\
\equiv & \biggl(\pop_{\menu_1=\menu_p\otimes \menu_q}
\frac{\prob_p(\menu_p)\prob_q(\menu_q)}{\beta}\biggr)\times \notag \\
& \ \ \ \ \times {\biggl(\pop_{p
\ptran{\pi_i}p_i,I(p_i)=\menu_p}\frac{\pi_i}{\prob_p(\menu_p)}
p_{i_{(\menu_p,a)}}\biggr)\po \biggr(\pop_{q
\ptran{\rho_j}q_j,I(q_j)=\menu_q}\frac{\rho_j}{\prob_q(\menu_q)}
q_{j_{(\menu_q,a)}}\biggr)} \notag \\
\equiv & \biggl(\pop_{ \menu_1=\menu_p\otimes \menu_q}
\frac{\prob_p(\menu_p)\prob_q(\menu_q)}{\beta
}\biggr)\times \notag \\
& \ \ \ \ \ \ \ \ \ \ \ \ \ \ \ \ \ \ \ \times
{\biggl(\pop_{\begin{scriptsize}\begin{array}{c}
                                 p \ptran{\pi_i}p_i,I(p_i)=\menu_p, \\
                                 q \ptran{\rho_j}q_j,I(q_j)=\menu_q
                               \end{array}\end{scriptsize}
}\frac{\pi_i \rho_j}{\prob_p(\menu_p)\prob_q(\menu_q)}
\Bigl(p_{i_{(\menu_p,a)}}\po
q_{j_{(\menu_q,a)}}\Bigr)\biggr)} \notag \\
\equiv&\pop_{\begin{scriptsize}\begin{array}{c}
               p \ptran{\pi_i}p_j,q\ptran{\rho_j}q_j, \\
               \menu_1=I(p_i)\otimes
I(q_j)
             \end{array}\end{scriptsize}
 }{\frac{\pi_i\rho_j}{\alpha} \Bigl(p_{i_{(I(p_i),a)}}\po
q_{j_{(I(q_j),a)}}}\Bigl).
\end{align}\\

From \eqref{eq:cong1} and \eqref{eq:cong2} we have

\begin{equation}\label{eq:cong3}
(p \po q)_{(\menu_1,a)} \equiv \pop_{\menu_p,\menu_q:
\menu_1=\menu_p\otimes \menu_q}
\frac{\prob_p(\menu_p)\prob_q(\menu_q)}{\sum_{\menu_p,\menu_q}{\prob_p(\menu_p)
\prob_q(\menu_q)}}{\Bigl(p_{(\menu_p,a)}\po q_{(\menu_q,a)}\Bigr)}.
\end{equation}

Similarly,
 \begin{equation}\label{eq:cong4}(\pbar \po q)_{(\menu_1,a)} \equiv
\pop_{\menu_p,\menu_q: \menu_1=\menu_p\otimes \menu_q}
\frac{\prob_{\pbar}(\menu_p)\prob_q(\menu_q)}{\sum_{\menu_p,\menu_q}{\prob_{\pbar}(\menu_p)
\prob_q(\menu_q)}}{\Bigl(\pbar_{(\menu_p,a)}\po
q_{(\menu_q,a)}\Bigr)}.
\end{equation}

From the inductive assumption and because $p \barbed \pbar$ and
$\barbed$ is congruence for $\pop$, we have

\begin{align}\label{eq:cong5}
&\pop_{\menu_p,\menu_q: \menu_1=\menu_p\otimes \menu_q}
\frac{\prob_p(\menu_p)\prob_q(\menu_q)}{\sum_{\menu_p,\menu_q}{\prob_p(\menu_p)
\prob_q(\menu_q)}}{\Bigl(p_{(\menu_p,a)}\po q_{(\menu_q,a)}\Bigr)} \notag \\
\equiv & \pop_{\menu_p,\menu_q: \menu_1=\menu_p\otimes \menu_q}
\frac{\prob_{\pbar}(\menu_p)\prob_q(\menu_q)}{\sum_{\menu_p,\menu_q}{\prob_{\pbar}(\menu_p)
\prob_q(\menu_q)}}{\Bigl(\pbar_{(\menu_p,a)}\po
q_{(\menu_q,a)}\Bigr)}.
\end{align}

 From \eqref{eq:cong3}, \eqref{eq:cong4}, and \eqref{eq:cong5}
 it follows that $(p \po q)_{(\menu_1,a)} \equiv$ $(\pbar
\po q)_{(\menu_1,a)} $.\\

\noindent \textbf{Case 1.2} \ $a_1 \not \in L$, $a_1$ appears in
$p$. The proof is similar to Case 1, with the difference that
instead of a process $q_{(\menu_q,a_1)}$, we use a process
$q_{(\menu_q)}$. The last one is defined by a process graph obtained
in a similar way as $q_{(\menu_q,a_1)}$, with the exception that
$q_{(\menu_q)}$ is
``ready'' to choose any action from the menu $\menu_q$.\\

\noindent \textbf{Case 1.3} \ $a_1 \not \in L$, $a_1$ appears in $q$
- symmetric to Case
2.\\

\noindent \textbf{Case 2} Suppose now that $\prob_{(p \po  q)}^{k}
(\menu_k | \menu_1, a_1, \ldots \menu_{k-1},a_{k-1})$ is defined but
$\prob_{(\pbar \po q)}^{k} (\menu_k | \menu_1, a_1, \ldots
\menu_{k-1},a_{k-1})$ is not defined. Either $\prob_{(p \po
q)}(\menu_1)>0$ while $\prob_{(\pbar \po q)}(\menu_1)=0$, which is
not possible because $p \barbed \pbar$, or $\prob_{(p \po
q)_{(\menu_1,a_1)}}^{k-1} (\menu_k | \menu_2, a_2, \ldots
\menu_{k-1},a_{k-1})$ is defined but $\prob_{(\pbar \po
q)_{(\menu_1,a_1)}}^{k-1} (\menu_k | \menu_2, a_2, \ldots
\menu_{k-1},a_{k-1})$ is not defined, which again is not possible
because of the inductive assumption.
\end{proof}
\end{LONGVERSION}

%
%
The following two theorems formulate the laws of distributivity of
the operators over probabilistic choice.

\begin{theorem}\label{thm:external} For processes $\{x_{ij}\}_{i \in I, j \in J}$ and actions $\{a_{i}\}_{i \in I} \subseteq \actset$, it holds $\extch_{i \in I}a_i.\pop_{j \in J}\pi_j x_{ij}\barbed \pop_{j \in
J}\pi_j\extch_{i \in I}a_i. x_{ij}$.

\end{theorem}
\begin{LONGVERSION}
\begin{proof}
Let $\menu=\{a_i\}_{i \in I}$, $p \equiv \extch_{i \in I}a_i.\pop_{j
\in J}\pi_j x_{ij}$ and $\pbar \equiv \pop_{j \in J}\pi_j\extch_{i
\in I}a_i. x_{ij}$. Then,  it is easy to show that, for every $i \in
I$, $p_{(\menu, a_i)} \barbed \pbar_{(\menu, a_i)}$. Let $n
> 1$ and  $(\menu_1,b_1, \ldots \menu_n)$ be an observation.  Then,

\[\prob_p^n(\menu_n|\menu_1,b_1, \ldots,  b_{n-1})=
\begin{cases}{}
 \prob_{p_{(\menu_1,b_1)}}^{n-1}(\menu_n|\menu_2,b_2, \ldots
 b_{n-1}) & \textrm{if } \menu_1=\menu , b_1 \in \menu
\\
\textrm{undefined} & \textrm{otherwise},
\end{cases}
\]
and
\[\prob_{\pbar} ^n(\menu_n|\menu_1,b_1, \ldots,  b_{n-1})=
\begin{cases}{}
 \prob_{{\pbar}_{(\menu_1,b_1)}}^{n-1}(\menu_n|\menu_2,b_2, \ldots
 b_{n-1}) & \textrm{if } \menu_1=\menu , b_1 \in \menu
\\
\textrm{undefined} & \textrm{otherwise}.
\end{cases}
\]

Now, it easily follows that  
$p \barbed \pbar$.
\end{proof}
\end{LONGVERSION}


\begin{theorem}\label{thm:distr} For every context $C[\cdot]$, it holds $C[\pop_{i \in I}\pi_i x_i]\barbed \pop_{i \in I}\pi_i C[ x_i]
$.
\end{theorem}
\begin{LONGVERSION}
\begin{proof}
By structural induction, similarly to the proof of Theorem
\ref{thm:external}.
\end{proof}
\end{LONGVERSION}

\section{Relationship between $\testing$ and
$\barbed$}\label{sec:relationship}

We establish our main result, namely that the testing equivalence
$\testing$ coincides with the probabilistic ready trace equivalence
$\barbed$. As an intermediate result, we prove that probabilistic
transitions do not add distinguishing power to the tests.
\begin{SHORTVERSION}The complete proofs of the results can be found
in \cite{testing_report}.\end{SHORTVERSION}

\begin{theorem}\label{thm:desno_levo}
Let $s$ and $t$ be two processes. If $s \barbed t$ then $s \testing
t$.
\end{theorem}
\begin{LONGVERSION}
\begin{proof}
Suppose $s \not \testing t$.  There exists a test $T$ such that
$\Result(s,T) \not = \Result(t,T)$. W.l.g., assume that $s$ and $t$
start with probabilistic transitions. By Def. \ref{def:Results},

\begin{equation}\label{eq:proof2_1}
 \Result(s,T)   =    \sum_{T\ptran{\rho_j}T_j} \rho_j \sum_{s \ptran{\pi_i}s_i}{\pi_i  \sum_{a
\in I(s_i)\cap I(T_j)}{\frac{a}
 {\sum_{b \in I(s_i)\cap I(T_j)}{b}} \cdot \Result (s_{ia},T_{ja}).}} 
 \end{equation}

\noindent By Def. \ref{def:pafter}, from \eqref{eq:proof2_1} we
obtain

\begin{align}\label{eq:proof2_a}
\Result(s,T) = & \sum_{\menu':\prob_T^1(\menu')>0}\prob_T^1(\menu') \sum_{\menu:\prob_s^1(\menu)>0}\prob_s^1(\menu) \times  \qquad \qquad \qquad \qquad \qquad \qquad \qquad \notag\\
   & \qquad \qquad \qquad \times  \sum_{a \in \menu \cap \menu'} \frac{a}{\sum_{b \in \menu
\cap \menu'}b}  \Result(s_{(\menu,a)}, T_{(\menu',a)}).
\end{align}

Similarly  we obtain

\begin{align}\label{eq:proof2_aa}
\Result(t,T) = & \sum_{\menu':\prob_T^1(\menu')>0}\prob_T^1(\menu') \sum_{\menu:\prob_t^1(\menu)>0}\prob_t^1(\menu) \times  \qquad \qquad \qquad \qquad \qquad \qquad \qquad \notag\\
   & \qquad \qquad \qquad \times  \sum_{a \in \menu \cap \menu'} \frac{a}{\sum_{b \in \menu
\cap \menu'}b}  \Result(t_{(\menu,a)}, T_{(\menu',a)}).
\end{align}

Now, assume $s \barbed t$. Define a \emph{length of a test} to be
the length of the longest sequence of actions the test can perform
before executing the action $\success$.
  The proof is by induction on the minimal length of
 a nonprobabilistic test that distinguishes between $s$ and $t$.

Let $T$ be a test of length $1$ such that $\Result(s,T) \not =
\Result(t,T)$. From Def. \ref{def:Results} it follows that for every
process $u$,
\begin{equation}\label{eq:proof2_b}
\Result(u, T_{(\menu,a)})=\prob^1_{T_{(\menu',a)}}(\{\success\}).
\end{equation}

From \eqref{eq:proof2_a} and \eqref{eq:proof2_b} we have

\begin{align}\label{eq:proof2_c}
\Result(s,T) = & \sum_{\menu':\prob_T^1(\menu')>0}\prob_T^1(\menu') \sum_{\menu:\prob_s^1(\menu)>0}\prob_s^1(\menu) \times  \qquad \qquad \qquad \qquad \qquad \qquad \qquad \notag\\
   & \qquad \qquad \qquad \qquad \qquad \qquad \times  \sum_{a \in \menu \cap \menu'} \frac{a}{\sum_{b \in \menu
\cap \menu'}b}  \prob^1_{T_{(\menu',a)}}(\{\success\}).
\end{align}

Similarly we obtain

\begin{align}\label{eq:proof2_d}
\Result(t,T) = & \sum_{\menu':\prob_T^1(\menu')>0}\prob_T^1(\menu') \sum_{\menu:\prob_t^1(\menu)>0}\prob_t^1(\menu) \times  \qquad \qquad \qquad \qquad \qquad \qquad \qquad \notag\\
   & \qquad \qquad \qquad \qquad \qquad \qquad  \times  \sum_{a \in \menu \cap \menu'} \frac{a}{\sum_{b \in \menu
\cap \menu'}b}  \prob^1_{T_{(\menu',a)}}(\{\success\}).
\end{align}

 From \eqref{eq:proof2_c},\eqref{eq:proof2_d} and from the assumption that $\prob_s^1(\menu) =\prob_t^1(\menu)$
for every menu $\menu$,
 we obtain that $\Result(s,T) = \Result(t,T)$, i.e. we obtain contradiction. Therefore,
 there exists a menu $\menu$ such that $\prob^1_s(\menu)\not =\prob^1_t(\menu)$, i.e. $s \not \barbed t$. \\

Let $T$ be a  test of length greater than one such that
$\Result(s,T) \not = \Result(t,T)$
 If there exists a
menu $\menu$ such that $\prob_s^1(\menu)\not =\prob_t^1(\menu) $,
then $s \not \barbed t$ and the proof is over. Therefore, suppose
$\prob_s^1(\menu) =\prob_t^1(\menu) $ for every menu $\menu
\subseteq \actset$. From \eqref{eq:proof2_a} and
\eqref{eq:proof2_aa} we have that for some menus $\menu$, $\menu'$
and action $a \in \menu \cap \menu'$, it holds
$\Result(s_{(\menu,a)}, T_{(\menu',a)})$ $\not
 =$ $\Result(t_{(\menu,a)}, T_{(\menu',a)})$. Now, by the inductive
 assumption, we have $s_{(\menu,a)}\not \barbed t_{(\menu,a)}$, i.e.
 there exists a ready trace $(\menu_2,a_2, \ldots \menu_k)$  such that $\prob_{s_{(\menu,a)}}^{k-1}(\menu_k | \menu_2, a_2, \ldots a_{k-1})$
  $\not =$ $\prob_{t_{(\menu,a)}}^{k-1}(\menu_k | \menu_2, a_2, \ldots
  a_{k-1})$ (or they are not defined at the same time). From the last,  from the assumption that  $\prob_s^1(\menu)=
  \prob_t^1(\menu)>0$, and from Def. \ref{def:P} it follows  that $\prob_{s}^{k}(\menu_k | \menu,a, \menu_2, a_2, \ldots a_{k-1})$
  $\not =$ $\prob_{t}^{k}(\menu_k | \menu,a, \menu_2, a_2, \ldots
  a_{k-1})$ (or they are not defined at the same time), i.e. $s \not \barbed
  t$. This completes the proof of the theorem.
\end{proof}
\end{LONGVERSION}

\begin{LONGVERSION}
The following lemma, which considers the determinant of a certain
type of an almost-triangular matrix, shall be needed in the proof of
Theorem \ref{thm:levo_desno}.

\begin{lemma} \label{lem:determinant}Let $\mathbf{Q}$ be a square $n \times n$ matrix
with elements $q_{ij}$, for $1\leq i \leq n$ and $1\leq j \leq n$.
Suppose $q_{ij}\in \{0,1\}$ for $i>1$, $q_{ij}=1$ for $i=j+1$,
$q_{ij}=0$ for $i>j+1$, and $q_{1j}=\frac{Q_1}{Q_j}$ for $1\leq j
\leq n$, where $Q_1, Q_2\ldots Q_n$ are  irreducible, mutually prime
polynomials with positive variables, and of non-zero degrees. Then
the determinant of $\mathbf{Q}$ is a non-zero rational function.
\end{lemma}

\begin{proof}
The determinant $\Det(\mathbf{Q})$ of  matrix $\mathbf{Q}$ can be
obtained from the general recursive formula
$\Det(\mathbf{Q})=\sum_{j=1}^{n}(-1)^{1+j}q_{1j}\Det(\mathbf{Q_{1j}})
$, where $\mathbf{Q_{1j}}$  is  the matrix obtained by deleting the
first row and the $j$-th column of $\mathbf{Q}$. Observe that
$\mathbf{Q_{1n}}$ is an upper-triangular matrix, the diagonal
elements of which are all equal to one. Since the determinant of a
triangular matrix is equal to the product of its diagonal elements,
we have $\Det(\mathbf{Q_{1n}})=1$. Therefore, the coefficient in
front of the rational function $\frac{Q_1}{Q_n}$ in
$\Det(\mathbf{Q})$ is equal to $1$. Suppose $\Det(\mathbf{Q})$ is a
zero-function. Then, the rational function $\frac{1}{Q_n}$ is equal
to a
 linear combination of $\frac{1}{Q_1}, \ldots \frac{1}{Q_{n-1}}$. This means  that the rational function $\frac{Q_1 \cdot
Q_2\cdot Q_{n-1}}{Q_n}$ is a polynomial. The last
 is impossible,
since, by assumption, the denominator is irreducible polynomial of
non-zero degree and is not contained in the numerator. Therefore,
$\Det(\mathbf{Q})$ is not a zero-function.
\end{proof}
\end{LONGVERSION}

\begin{theorem} \label{thm:levo_desno}
 Let $s$ and $t$ be two processes such that $s \not \barbed
t$. There exists a test $T$ that has no probabilistic transitions
such that $\Result(s,T) \not = \Result(t,T)$.
\end{theorem}

\begin{proof}\begin{SHORTVERSION} (Outline)
The proof is technically rather involved \cite{testing_report},
which is why here we give only the main lines of it.
\end{SHORTVERSION} We prove the theorem by induction on the minimal
length $m$ of a ready trace that distinguishes between $s$ and $t$.
For $m=1$, we prove that the test $T=\extch_{a\not \in
\menu}{a.\success}$, where $\menu$ is a menu with a minimal possible
number of actions such that $\prob_s^1(\menu)\not =
\prob_t^1(\menu)$, distinguishes between $s$ and $t$. For $m>1$ the
proof goes as follows. If $\prob_s^1(\menu) = \prob_t^1(\menu)$ for
every menu $\menu$, then by the inductive assumption it follows that
there exists a test $T_1$, menu $\menu_1$ and action $a_1 \in
\menu_1$ such that $\Result(s_{(\menu_1,a_1)},T_1)
 \not  = \Result(t_{(\menu_1,a_1)},T_1)$. We show that there exists
 a subset of the action set, say $\mathsf{Act}$, such that the test $T=a_1.T_1+\extch_{b \in
 \mathsf{Act}}.\success$ distinguishes between $s$ and $t$. To prove this,
we take $\menu_1$ to be
 the menu containing a minimal possible number of actions such that $\prob_s^1(\menu_1)>0$, $a_1 \in \menu_1$,  and
$\Result(s_{(\menu_1,a_1)},T_1)
 \not  = \Result(t_{(\menu_1,a_1)},T_1)$. Then we take the set
 $\mathsf{Act'}$ to consist of the actions that can be initially
 performed by $s$ but do not belong to menu $\menu_1$. Then, we show
 that there must exist a subset $\mathsf{Act}$ of $\mathsf{Act'}$
 such that the test $T=a_1.T_1+\extch_{b \in
 \mathsf{Act}}.\success$ distinguishes between $s$ and $t$ (otherwise, we obtain that $\Result(s_{(\menu_1,a_1)},T_1)
   = \Result(t_{(\menu_1,a_1)},T_1)$, which contradicts our assumption).

\begin{LONGVERSION}
We now proceed with a detailed presentation of the proof.

 From $s \not \barbed t$ and by Def. \ref{def:barbed}, there
must exist a ready trace $(\menu_1, a_1, \ldots \menu_m)$ such that
$\prob_s^m(\menu_m | \menu_1, a_1, \ldots a_{m-1}) \not =
\prob_t^m(\menu_m | \menu_1, a_1, \ldots a_{m-1})$. The proof is by
induction on $m$.\\

\textbf{Case 1 ($m=1$)} Suppose first that there exists a menu
$\menu$ such that
$\prob_s^1(\menu)\not = \prob_t^1(\menu)$. 
Let $\menu$ be a menu with a minimal possible number of actions such
that $\prob_s^1(\menu)\not = \prob_t^1(\menu)$. Take
$T=\extch_{a\not \in \menu}{a.\success}$. We have $\Result(s,T) = 1-
\sum_{\menu' \subseteq \menu}\prob_s^1(\menu')$, because the actions
of $s$ and $T$ will fail to synchronize if and only if the random
choice decides that menu $\menu$ or some menu $\menu' \subset \menu$
is offered to process $s$ initially. Similarly, $\Result(t,T) = 1-
\sum_{\menu' \subseteq \menu}\prob_t^1(\menu')$. Now, suppose that
$\Result(s,T)
  = \Result(t,T)$. We have
   $\sum_{\menu' \subseteq \menu}\prob_s^1(\menu')=
   \sum_{\menu' \subseteq \menu}\prob_t^1(\menu')$. From this and  $\prob_s^1(\menu)\not =
\prob_t^1(\menu)$, it follows that there  exists a menu $\menu'
\subset \menu$ such that also $\prob_s^1(\menu')\not =
\prob_t^1(\menu')$. But this contradicts the assumption that $\menu$
is a menu with a minimal possible number of actions such that
$\prob_s^1(\menu)\not =
\prob_t^1(\menu)$.\\

\textbf{Case 2 ($m>1$)} Suppose now that $\prob_s^1(\menu) =
\prob_t^1(\menu)$ for every menu $\menu$. Let $(\menu_1, a_1, \ldots
\menu_m)$ be a ready trace
such that 
$\prob_s^{m-1}(\menu_m | \menu_1, a_1, \ldots a_{m-1}) \not =
\prob_t^{m-1}(\menu_m | \menu_1, a_1, \ldots a_{m-1})$. From
$\prob_s^1(\menu_1) = \prob_t^1(\menu_1)$, and from  Definitions
\ref{def:pafter} and \ref{def:P}, it follows that
$\prob_{s_{(\menu_1,a_1)}}^{m-1}(\menu_m | \menu_2, a_2, \ldots
a_{m-1}) \not = \prob_{t_{(\menu_1,a_1)}}^{m-1}(\menu_m | \menu_2,
a_2, \ldots a_{m-1})$ (in case $m=2$,
$\prob_{s_{(\menu_1,a_1)}}^1(\menu_2 ) \not =
\prob_{t_{(\menu_1,a_1)}}^1(\menu_2) $). Now, by the inductive
assumption, there exists a non-probabilistic test $T_1$ such that
$\Result(s_{(\menu_1,a_1)},T_1)
 \not  = \Result(t_{(\menu_1,a_1)},T_1)$.\\

\textbf{Case 2.1} \ Suppose first that $a_1$ does not belong to any
first-level menu of $s$ other than $\menu_1$, i.e. that for every
menu $\menu$, $\prob_s^1(\menu)>0$ and $a_1 \in \menu$ implies
$\menu = \menu_1$.  Then the test $T=a_1.T_1$ distinguishes
between $s$  and $t$.\\

\textbf{Case 2.2} \ Suppose now that $a_1$ belongs to at least one
first-level menu of
 $s$ other than $\menu_1$, i.e.
there exists at least one menu $\menu \not = \menu_1$ such that
$\prob_s^1(\menu)>0$ and $a_1 \in \menu$.
 Without loss of
 generality, assume that $\menu_1$ is a menu with a minimal possible number of actions such that
$\prob_s^1(\menu_1)>0$, $a_1 \in \menu_1$,  and
$\Result(s_{(\menu_1,a_1)},T_1)
 \not  = \Result(t_{(\menu_1,a_1)},T_1)$.
 Let $\{b_j\}_{j \in J}$ be
the set of actions that appear in the first level of $s$ (and
therefore $t$) but not in $\menu_1$, i.e. $b \in \{b_j\}_{j \in J}$
if and only if $b \not \in \menu_1$ and there exists a menu $\menu$
such that $\prob_s^1(\menu)>0$, $b \in \menu$. We shall prove that
there exists $J' \subseteq J$ such  that the test $T=a_1.T_1+
\extch_{j \in J' }b_j.\success$ distinguishes between $s$ and $t$.
More concretely, we shall prove that, assuming the opposite, it
follows that $\Result(s_{(\menu_1,a_1)},T_1)
  = \Result(t_{(\menu_1,a_1)},T_1)$, thus obtaining contradiction.
  \\

\textbf{Case 2.2.a} \ Suppose first that $\{b_j\}_{j \in J} =
\emptyset$. This means that there are no actions other than those in
$\menu_1$, that appear in the first level of $s$. Therefore, all
menus $\menu$ for which $\prob_s^1(\menu)>0$ satisfy $\menu
\subseteq \menu_1$. We prove that the test $T=a_1.T_1$ distinguishes
between $s$ and $t$. Assume that $\Result(s,T)=\Result(t,T)$. From
the last and from Def. \ref{def:Results}, we obtain

\begin{equation}\label{eq:submenus}
\sum_{\menu:\prob_s^1(\menu)>0,a_1\in\menu \subseteq
\menu_1}(\Result (s_{(\menu,a_1)},T_1)-\Result
(t_{(\menu,a_1)},T_1))=0.
\end{equation}
 By
assumption, for every $\menu \subset \menu_1$ such that $a_1 \in
\menu$ it holds $\Result (s_{(\menu,a_1)},T_1)=\Result
(t_{(\menu,a_1)},T_1)$. Therefore, from \eqref{eq:submenus} we
obtain $\Result (s_{(\menu_1,a_1)},T_1)=\Result
(t_{(\menu_1,a_1)},T_1)$, which contradicts the assumption $\Result
(s_{(\menu_1,a_1)},T_1)\not =\Result (t_{(\menu_1,a_1)},T_1)$.\\

\textbf{Case 2.2.b} \ Suppose now that $\{b_j\}_{j \in J} \not =
\emptyset$.
 Given action
$b_i \in \{b_j\}_{j \in J}$, denote by $\mathcal{M}_i$ the set of
all first-level menus of $s$ that contain $b_i$ and $a_1$, i.e.
$\menu \in \mathcal{M}_i$ iff $\prob_s^1(\menu) >0$ and $b_i,a_1 \in
\menu$; denote by $\mathcal{M}_i^C$ the set of all first-level menus
of $s$ that do not contain $b_i$ but have $a_1$, i.e. $\menu \in
\mathcal{M}_i^C$ iff $\prob_s^1(\menu) >0$, $b_i \not \in \menu$ and
$a_1 \in \menu$.

Let $T=a_1.T_1+\extch_{j \in J' }b_j.\success$ for some
$J'=\{1,2,\ldots n\} \subseteq J$ and suppose $\Result(s,T)
  = \Result(t,T)$. Since $\prob_s^1(\menu)=\prob_t^1(\menu)$ for every menu $\menu$,
  observe that only if action $a_1$ is performed initially, it is
possible for the test $T=a_1.T_1+\extch_{j \in J' }b_j.\success$  to
make a difference between $s$ and $t$. Because of this and by
Definitions \ref{def:Results} and \ref{def:pafter}  it follows  that

 \begin{align}\label{eq:thebig}
&\sum_{\menu \in \mathcal{M}_n^C \cap \mathcal{M}_{n-1}^C \cap
\cdots \cap
\mathcal{M}_1^C}\frac{a_1}{a_1}\prob_s^1(\menu)(\Result(s_{(\menu,a_1)},T_1)-\Result(t_{(\menu,a_1)},T_1))  \notag\\
+&\sum_{\menu \in \mathcal{M}_n^C \cap \mathcal{M}_{n-1}^C \cap
\cdots \cap
\mathcal{M}_1}\frac{a_1}{a_1+b_1}\prob_s^1(\menu)(\Result(s_{(\menu,a_1)},T_1)-\Result(t_{(\menu,a_1)},T_1))  \notag\\
+&\cdots \notag\\
+&\sum_{\menu \in \mathcal{M}_n   \cap \cdots \cap
\mathcal{M}_1}\frac{a_1}{a_1+\sum_{j=1}^{n}b_j}\prob_s^1(\menu)(\Result(s_{(\menu,a_1)},T_1)-\Result(t_{(\menu,a_1)},T_1))\notag\\
=& \quad 0.
\end{align}

\noindent Each intersection appearing under the $\sum$-operators of
\eqref{eq:thebig} can be mapped bijectively to a binary number of
$n$ digits -- the i-th digit being $0$ if the intersection contains
$\mathcal{M}_{n+1-i}^C$, and $1$ if the intersection contains
$\mathcal{M}_{n+1-i}$. (For reasons that will become clear later,
the order of the indexing is reversed.)

Suppose  $\Result(s,T) = \Result(t,T)$ for every test
$T=a_1.T_1+\extch_{j \in J' }b_j.\success$, where $J'\subseteq J$.
We shall prove that, in this case, every sum
$\sum(\Result(s_{(\menu,a_1)},T_1)-\Result(t_{(\menu,a_1)},T_1))$
that appears in \eqref{eq:thebig} when $J'=J$ is equal to a
zero-function. In particular, the equality

\begin{equation}\label{eq:presek}
\sum_{\menu \in \bigcap_{j \in
J}\mathcal{M}_j^C}(\Result(s_{(\menu,a_1)},T_1)-\Result(t_{(\menu,a_1)},T_1))=0
\end{equation}

\noindent will hold. Note  that the set $\bigcap_{j \in
J}{\mathcal{M}_j^C}$ contains all first-level menus of $s$ that have
the action $a_1$ but do not have any other action that does not
appear in $\menu_1$. Therefore, $\bigcap_{j \in J}{\mathcal{M}_j^C}$
consists of the subsets of $\menu_1$ that contain $a_1$. Thus, the
equation \eqref{eq:presek} is equivalent to the equation
\eqref{eq:submenus} which leads to $\Result(s_{(\menu_1,a_1)},T_1)
=\Result(t_{(\menu_1,a_1)},T_1)$, i.e. to  contradiction. This would complete the proof of the theorem.\\

We now proceed with proving the above stated claim. We prove a more
general result,  namely that for $J' \subseteq J$, under assumption
that $\Result(s,T) = \Result(t,T)$ for every test
$T=a_1.T_1+\extch_{i \in J'' }b_i.\success$ such that $J'' \subseteq
J$ and $|J''|\leq|J'|$, it holds that every sum
$\sum(\Result(s_{(\menu,a_1)},T_1)-\Result(t_{(\menu,a_1)},T_1))$
that appears in \eqref{eq:thebig} is equal to zero. \\

Suppose first that $|J'|=1$, i.e. $J'=\{1\}$. Assume that
\begin{equation}\label{eq:eden1}
\Result(s,a_1.T_1) = \Result(t,a_1.T_1)
\end{equation}
  and
\begin{equation}\label{eq:eden2}
\Result(s,a_1.T_1+b_1.\success) = \Result(t,a_1.T_1+b_1.\success).
\end{equation}
 From \eqref{eq:eden1}, Def. \ref{def:Results}, and because $\prob_s^1(\menu)=\prob_t^1(\menu)$ for every menu $\menu$, we  obtain

\begin{align}\label{eq:b1}
\sum_{\menu \in \mathcal{M}_1 \cup \mathcal{M}_1^C
}\frac{a_1}{a_1}\prob_s^1(\menu)(\Result(s_{(\menu,a_1)},T_1)-\Result(t_{(\menu,a_1)},T_1))=0.
\end{align}
The equation \eqref{eq:thebig} turns into
\begin{align}\label{eq:b2}
& \sum_{\menu \in \mathcal{M}_1^C}\frac{a_1}{a_1}\prob_s^1(\menu)(\Result(s_{(\menu,a_1)},T_1)-\Result(t_{(\menu,a_1)},T_1))\notag\\
+ & \sum_{\menu \in
\mathcal{M}_1}\frac{a_1}{a_1+b_1}\prob_s^1(\menu)(\Result(s_{(\menu,a_1)},T_1)-\Result(t_{(\menu,a_1)},T_1))\notag
\\
= & \quad 0.
\end{align}

Denote $\sum_{\menu \in
\mathcal{M}_1^C}\prob_a^1(\menu)(\Result(s_{(\menu,a_1)},T_1)-\Result(t_{(\menu,a_1)},T_1))$
by $x_0$ and $\sum_{\menu \in
\mathcal{M}_1}\prob_a^1(\menu)(\Result(s_{(\menu,a_1)},T_1)-\Result(t_{(\menu,a_1)},T_1))$
by $x_1$. Our goal is to show  that $x_0=0$ and $x_1=0$, i.e. that
they are zero-functions.  From \eqref{eq:b1} and \eqref{eq:b2} we
obtain the following system of equations for the unknowns $x_0$ and
$x_1$:
$$
\begin{cases}{}
\frac{a_1}{a_1}x_0  +  \frac{a_1}{a_1+b_1}x_1  =0\\
x_0  +  x_1  =0,
\end{cases}
$$

 \noindent or in a matrix form $$ \mathbf{Q_1}\mathbf{x}=\mathbf{0},$$ where
$$\mathbf{Q_1}=\left(
                      \begin{array}{cc}
                        \frac{a_1}{a_1} & \frac{a_1}{a_1+b_1} \\
                        1 & 1 \\
                      \end{array}
                    \right), \mathbf{x}=\left(
                                            \begin{array}{c}
                                              x_1 \\
                                              x_2 \\
                                            \end{array}
                                          \right)
                    , \textrm{ and } \mathbf{0}=\left(
                                         \begin{array}{c}
                                           0 \\
                                           0 \\
                                         \end{array}
                                       \right).
                    $$
Since the determinant of the matrix $\mathbf{Q_1}$ is not a
zero-function, it follows that $x_0=0$ and $x_1=0$ is the only
solution of
the system.\\

To present a better intuition on the proof in the general case, we
shall also consider separately  the case $|J'|=2$. Let $J'=\{1,2\}$
and assume that $\Result(s,T) = \Result(t,T)$ for every test
$T=a_1.T_1+\extch_{i \in J'' }b_i.\success$ such that $J'' \subseteq
J$ and $|J''|\leq|J'|$. The equation \eqref{eq:thebig} turns into

 \begin{align}\label{eq:thebig2}
& \sum_{\menu \in \mathcal{M}_2^C \cap \mathcal{M}_1^C}\frac{a_1}{a_1}\prob_s^1(\menu)(\Result(s_{(\menu,a_1)},T_1)-\Result(t_{(\menu,a_1)},T_1))  \notag\\
+& \sum_{\menu \in \mathcal{M}_2^C \cap \mathcal{M}_1}\frac{a_1}{a_1+b_1}\prob_s^1(\menu)(\Result(s_{(\menu,a_1)},T_1)-\Result(t_{(\menu,a_1)},T_1))  \notag\\
+& \sum_{\menu \in \mathcal{M}_2 \cap \mathcal{M}_1^C}\frac{a_1}{a_1+b_2}\prob_s^1(\menu)(\Result(s_{(\menu,a_1)},T_1)-\Result(t_{(\menu,a_1)},T_1))  \notag\\
+& \sum_{\menu \in \mathcal{M}_2 \cap \mathcal{M}_1}\frac{a_1}{a_1+b_1+b_2}\prob_s^1(\menu)(\Result(s_{(\menu,a_1)},T_1)-\Result(t_{(\menu,a_1)},T_1))\notag\\
=& \qquad 0.
\end{align}

Denoting   $\sum_{\menu \in \mathcal{M}_2^C \cap \mathcal{M}_1^C
}\prob_s^1(\menu)(\Result(s_{(\menu,a_1)},T_1)-\Result(t_{(\menu,a_1)},T_1))$
by $x_{00}$ and so on, \eqref{eq:thebig2} turns into
\begin{align}
\frac{a_1}{a_1}x_{00} + \frac{a_1}{a_1+b_1}x_{01} +
\frac{a_1}{a_1+b_2} x_{10} + \frac{a_1}{a_1+b_1+b_2} x_{11}=0.
\end{align}

From $\sum_{\menu \in
\mathcal{M}_2^C}\prob_s^1(\menu)(\Result(s_{(\menu,a_1)},T_1)-\Result(t_{(\menu,a_1)},T_1))=0$
we obtain $x_{00}+x_{01}=0$, and from $\sum_{\menu \in
\mathcal{M}_2}\prob_s^1(\menu)(\Result(s_{(\menu,a_1)},T_1)-\Result(t_{(\menu,a_1)},T_1))=0$
we obtain $x_{10}+x_{11}=0$. Similarly, from $\sum_{\menu \in
\mathcal{M}_1}\prob_s^1(\menu)(\Result(s_{(\menu,a_1)},T_1)-\Result(t_{(\menu,a_1)},T_1))=0$
we obtain that $x_{01}+x_{11}=0$. Therefore, we have the following
system of equations:

$$
\begin{cases}{}
\frac{a_1}{a_1}x_{00} + \frac{a_1}{a_1+b_1}x_{01} +
\frac{a_1}{a_1+b_2} x_{10} + \frac{a_1}{a_1+b_1+b_2} x_{11}=0\\
x_{00}+x_{01}=0\\
x_{01}+x_{11}=0 \\
 x_{10}+x_{11}=0.
\end{cases}
$$

The main matrix of the system is $$\mathbf{Q_2}= \left(
  \begin{array}{cccc}
    \frac{a_1}{a_1} & \frac{a_1}{a_1+b_1} & \frac{a_1}{a_1+b_2} & \frac{a_1}{a_1+b_1+b_2} \\
    1 & 1 & 0 & 0 \\
    0 & 1 & 0 & 1 \\
    0 & 0 & 1 & 1 \\
  \end{array}
\right). $$

By Lemma \ref{lem:determinant}, $\Det(\mathbf{Q_2})$ is not a
zero-function, which implies that the vector of zero-functions is
the only solution of the above system of equations.\\

We now present how each matrix $\mathbf{Q_{n+1}}$ can be obtained
from the matrix $\mathbf{Q_{n}}$.

In general, for $\mathcal{M}_i^* \in
\{\mathcal{M}_i,\mathcal{M}_i^C\}$, it holds

\begin{align}
&\sum_{\menu \in (\bigcap_{i=1}^n{\mathcal{M}_i^*}) \cap
\mathcal{M}_{n+1}}
\prob_s^1(\menu)(\Result(s_{(\menu,a_1)},T_1)-\Result(t_{(\menu,a_1)},T_1))
\notag \\
+& \sum_{\menu \in (\bigcap_{i=1}^n{\mathcal{M}_i^*}) \cap
\mathcal{M}_{n+1}^{C}}
\prob_s^1(\menu)(\Result(s_{(\menu,a_1)},T_1)-\Result(t_{(\menu,a_1)},T_1))
\notag \\
=& \ \ \ \ \sum_{\menu \in (\bigcap_{i=1}^n{\mathcal{M}_i^*})}
\prob_s^1(\menu)(\Result(s_{(\menu,a_1)},T_1)-\Result(t_{(\menu,a_1)},T_1)).
\end{align}

This means that, in the general case, each solution $x_{i_1 i_2
\ldots i_n}=0$ of the system $\mathbf{Q_n}\mathbf{x}=\mathbf{0}$
generates  the following equations for the next system: \[x_{i_1 i_2
\ldots i_{k} 0 i_{k+1} \ldots i_{n}} + x_{i_1 i_2 \ldots i_{k} 1
i_{k+1} \ldots
i_{n}}=0,\] for every $0\leq k \leq n$. 
For example, in case $|J'|=3$ we obtain the following matrix:
$$\mathbf{Q_3}=\left(
              \begin{array}{cccccccc}
                \frac{a_1}{a_1} & \frac{a_1}{a_1+b_1} & \frac{a_1}{a_1+b_2} & \frac{a_1}{a_1+b_1+b_2} & \frac{a_1}{a_1+b_3} & \frac{a_1}{a_1+b_1+b_3} & \frac{a_1}{a_1+b_2+b_3} & \frac{a_1}{a_1+b_1+b_2+b_3} \\
                1 & 1 & 0 & 0 & 0 & 0 & 0 & 0 \\
                0 & 1 & 0 & 1 & 0 & 0 & 0 & 0 \\
                0 & 0 & 1 & 1 & 0 & 0 & 0 & 0 \\
                0 & 0 & 0 & 1 & 0 & 0 & 0 & 1 \\
                0 & 0 & 0 & 0 & 1 & 1 & 0 & 0 \\
                0 & 0 & 0 & 0 & 0 & 1 & 0 & 1 \\
                0 & 0 & 0 & 0 & 0 & 0 & 1 & 1 \\
              \end{array}
            \right).
$$
Note that each row of $\mathbf{Q_3}$, except the first one, contains
exactly two $1$'s, at positions whose binary representations differ
in exactly one place (for example at the positions $001$ and $011$).

Informally, the general algorithm for obtaining the elements
$q_{n+1}^{ij}$ of a $2^{n+1}\times 2^{n+1}$ matrix
$\mathbf{Q}_{n+1}$ from matrix $\mathbf{Q}_n$, assuming
$\mathbf{Q}_n$ is non-singular, is as follows. First, initialize all
elements of $\mathbf{Q}_{n+1}$ to zero. Then, copy $\mathbf{Q}_{n}$
into the upper left corner of $\mathbf{Q}_{n+1}$. Then, copy
$\mathbf{Q}_{n}$, excluding the first row, into the lower right
corner of $\mathbf{Q}_{n+1}$. Then, assign $1$ to $q_{n+1}^{ij}$ for
$i=2^n+1$ and $j \in \{2^n, 2^{n+1}\}$. Finally, add the appropriate
new rational fractions in the second half of the first row of
$\mathbf{Q}_{n+1}$. The key observation is that in this way, we
obtain again a matrix such that each row, except the first one,
contains exactly two $1$'s, at positions whose binary
representations differ in exactly one place. Formally,
\[ q_{n+1}^{ij}=
\begin{cases}{}
q_n^{ij} & \textrm{if } 1 \leq i \leq 2^n \textrm{ and } j \leq 2^n, \\
1 & \textrm{if } i=2^n+1 \textrm{ and } j \in \{2^n, 2^{n+1}\}, \\
q_n^{ij} & \textrm{if } 2^n+1<i \textrm{ and } 2^n <j, \\
\frac{a_1}{a_1+\sum_{k \in K}{b_k}+b_{n+1}} & \textrm{if } i=1,j>2^n, \textrm{ and }  q_n^{(i) (j-2^n)}= \frac{a_1}{a_1+\sum_{k \in K}{b_k}}\\
0 & \textrm{otherwise}.
\end{cases}
\]

Assuming matrix $\mathbf{Q}_{n}$ satisfies the conditions of Lemma
\ref{lem:determinant}, it easily follows that  matrix
$\mathbf{Q}_{n+1}$ also satisfies the conditions of Lemma
\ref{lem:determinant}. Therefore, its determinant is  not a zero
function. This means that the system
$\mathbf{Q}_{n+1}\mathbf{x}=\mathbf{0}$ has only zero-functions as
solutions, which we were aiming to prove. Therefore, the proof of
the theorem is complete.
\end{LONGVERSION}
\end{proof}

From Theorems \ref{thm:desno_levo} and \ref{thm:levo_desno} the
following statements directly follow.
\begin{corollary} For arbitrary processes $s$ and $t$,  $s \testing t$
if and only if $s \barbed t$.
\end{corollary}

\begin{corollary} For arbitrary processes $s$ and $t$,  $s \not \testing
t$ if and only if there exists a test $T$ without probabilistic
transitions such that $\Result(s,T)\not = \Result(t,T)$.
\end{corollary}

\section{Conclusion, future work, and related work }\label{sec:conclusion}
 \paragraph{Concluding remarks} We have proposed a testing equivalence in the style of \cite{DH84} for
processes where the internal nondeterminism is quantified with
probabilities. The testing semantics allows distribution of external
choice over probabilistic choice, i.e. accomplishes unobservability
of the internal probabilistic choice. The definition exploits a new
method for labeling the synchronized actions using rational
functions over the action labels, which, we believe, is of
independent interest. We have also developed an alternative
characterization of the testing equivalence, namely as a
probabilistic version of the ready trace equivalence
~\cite{Pnueli85, BBKready}. The definition of the latter uses
Bayesian probability. It is intuitive and can be easily justified by
a black box testing scenario akin to those in \cite{spectrum1,
CSV07}. We have also shown that it is congruence for all standard
operators for the given model, including asynchronous parallel
composition and priority.

\paragraph{Internal nondeterminism} It can be anticipated by now
that combining internal choice, probabilistic choice and parallel
composition is challenging. Again  ``cloning'' the internal
nondeterminism after the probabilistic choice in a parallel context
can ``erase'' the probabilities, which disallows distribution of
prefix over probabilistic choice (this phenomenon has been also
studied in ~\cite{CP07,CLSV06, GRS07, Lowe93, GD09,segalaPhD}).
\begin{figure}[t]\centering \large
$\xymatrix@R=0.3cm@C=0.1cm{
& X \ar@{-->}[dl]_{\frac{1}{2}} \ar@{-->}[dr]^{\frac{1}{2}} \\ 
          \nst\ar[d]_{wrt}  && \nst\ar[d]_{wrt}  \\
          \nst\ar[d]_{rev}  && \nst\ar[d]_{rev} \\
          \nst\ar[d]_{head}  && \nst\ar[d]_{tail}  \\
          \nst && \nst \\ 
           }$
 \quad $\po$ \quad
$\xymatrix@R=0.3cm@C=0.1cm{
&  Y \ar[d]_{wrt} \\
& \nst \ar[dl]_{} \ar[dr]^{} \\ 
          \nst\ar[d]_{rev}  && \nst\ar[d]_{rev}  \\
          \nst\ar[d]_{head}  && \nst\ar[d]_{tail} \\
          \nst\ar[d]_{\smiley} && \nst\ar[d]_{\smiley} \\ 
          \nst && \nst \\ 
           }
           $
           \quad $\longrightarrow$ \quad
  $
\xymatrix@R=0.3cm@C=0.1cm{
    & & & \nst  \ar@{-->}[dll]_{\frac{1}{2}} \ar@{-->}[drr]^{\frac{1}{2}} \\
    & \nst \ar[d]_{wrt}  & & & & \nst  \ar[d]_{wrt}  \\
& \nst \ar[dl]_{} \ar[dr]^{} & & & & \nst  \ar[dl]_{} \ar[dr]^{} \\ 
          \nst\ar[d]_{rev}  && \nst\ar[d]_{rev}  & & \nst\ar[d]_{rev} && \nst\ar[d]_{rev} \\ 
          \nst\ar[d]_{head}  && \nst  & & \nst && \nst\ar[d]_{tail} \\
          \nst\ar[d]_{\smiley} &&  & &  && \nst\ar[d]_{\smiley} \\
          \nst &&  & &  && \nst \\
           }$
\newline
 \caption{Synchronized coin tosser($X$) and  result-guesser($Y$) } \hbox{}\vspace{-0.5cm}
 \label{fig:game1}
\end{figure}
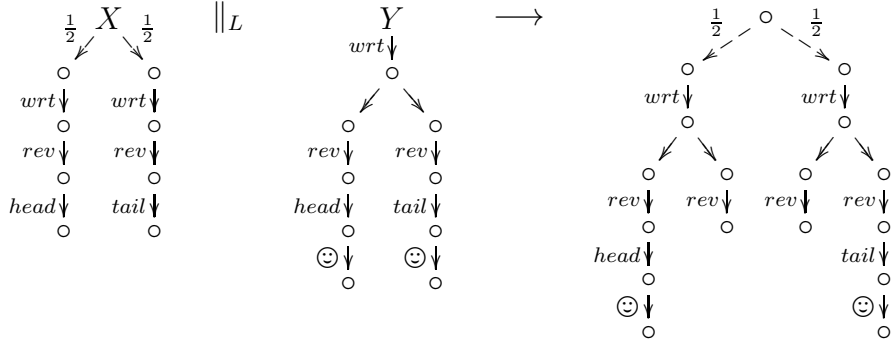
Namely, consider the following game. The player $X$ tosses a fair
coin and hides the outcome. Player $Y$ guesses the outcome of the
tossing and writes it down. While he is writing down the result,
player $X$ waits (i.e. he may write down something meaningless).
Then they both agree to reveal their outcomes, i.e. $X$ to uncover
the coin and $Y$ to show what he/she has written.\footnote{Note the
difference between this game and the example in Sec.
\ref{sec:introduction} -- in the former there is no external choice
in the original processes, while in the latter they don't have
internal nondeterminism.} Obviously, the probability that the second
player has guessed correctly equals $\frac{1}{2}$. However, the
resulting graph for the synchronization of both players (Fig.
\ref{fig:game1}) suggests that there is a strategy such that player
$Y$ can always guess the correct result. On  the other hand, if
process $\bar{X}=wrt.rev.(head \oplus_{\frac{1}{2}} tail)$ is
synchronized with $Y$, the resulting graph suggests that the
probability of reporting a $\smiley$ action is exactly
$\frac{1}{2}$. This prevents equating processes $X$ and $\bar{X}$,
i.e. allowing distribution of prefix over internal probabilistic
choice.
 Indeed, in
presence of internal nondeterminism, the testing equivalence of
\cite{YL92} and its variants have all been characterized as
simulations \cite{JY02,LSV07,DGHM08}. The proposed solutions
\cite{CP07,CLSV06, GRS07, GD09} to the problem with parallel
composition
 suggest that the process composition
needs to ``remember'' the outcome of the internal choice that a
component makes locally.
 To solve the problem in our setting in the
lines of these solutions, we  also plan to enrich the internal
transitions with labels that cannot communicate. Before composing
all labels would be different. If the original process has, for
example, two outgoing internal transitions labeled with $l_1$ and
$l_2$, then the composed process shall have transitions labeled with
$\frac{l_1}{l_1+l_2}$ and $\frac{l_2}{l_1+l_2}$.  Fig.
\ref{fig:game2} presents the result of testing process $X$ of Fig.
\ref{fig:game1} with process $Y$, assuming the internal transitions
of $Y$ are labeled with $l_1$ and $l_2$. Two processes would not be
distinguished by a test if both results of testing are equal modulo
isomorphism on the labels set. However, we leave the formal
definition of this testing semantics for future work.

\paragraph{Related Work}
\begin{wrapfigure}{r}{48mm}\centering
\large \hbox{}\vspace{-5mm} $ \xymatrix@R=0.3cm@C=0.2cm{
    & & & \nst  \ar@{-->}[dll]_{\frac{1}{2}} \ar@{-->}[drr]^{\frac{1}{2}} \\
    & \nst \ar[d]_{\frac{w}{w}}  & & & & \nst  \ar[d]_{\frac{w}{w}}  \\
& \nst \ar[dl]_{\frac{l_1}{l_1+l_2}} \ar[dr]^{\frac{l_2}{l_1+l_2}} & & & & \nst  \ar[dl]_{\frac{l_1}{l_1+l_2}} \ar[dr]^{\frac{l_2}{l_1+l_2}} \\ 
           \nst\ar[d]_{\frac{r}{r}}  && \nst\ar[d]_{\frac{r}{r}}  & & \nst\ar[d]_{\frac{r}{r}} && \nst\ar[d]_{\frac{r}{r}} \\
          \nst\ar[d]_{\frac{h}{h}}  &&   & & && \nst\ar[d]_{\frac{t}{t}} \\
           \nst\ar[d]_{\smiley} &&  & &  && \nst\ar[d]_{\smiley} \\
           \nst &&  & &  && \nst
           }$
\newline
 \caption{Testing with internal transitions.} \hbox{}\vspace{-0.5cm}
 \hbox{}\vspace{-1cm}
 \label{fig:game2}
\end{wrapfigure}
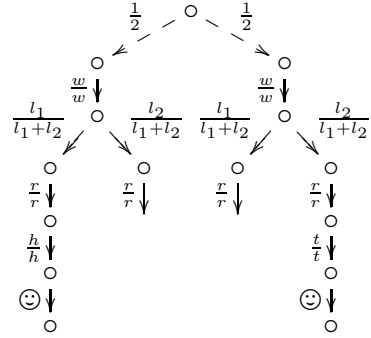
Process equivalences  that allow distribution of prefix over
probabilistic choice (i.e. unobservability of the random choice)
have been a research topic ever since probabilities were introduced
in concurrency theory (see
e.g.\cite{Lowe93,Seidel95,segalaPhD,morgan96refinementoriented,KN98b,CCVPP03,CP07,CSV07}).
However,  only \cite{Lowe93}, \cite{Seidel95}, and, under certain
conditions, \cite{CP07}, also allow distribution of external choice
over probabilistic, i.e. equate  processes $s$ and $\bar{s}$ of Fig.
\ref{fig_machine}. In~\cite{Lowe93} probabilistic versions of broom
(ready/failure) and barbed (ready/failure trace) equivalences are
defined. These definitions use ``probability functions'' that
 compute the maximal probability for a ready trace to occur
(i.e. they do not generate probability spaces over the set of ready
traces), which makes it hard to construct corresponding ``black-
box'' testing scenarios. In~\cite{Seidel95}, in the model with
external choice, a process is defined as conditional probability
measure over sequences of actions. This semantics also  identifies
processes $(a+b) \oplus_{\frac{1}{2}} c$ and $(a+c)
\oplus_{\frac{1}{2}} b$. Obviously, this is  not desirable.
In~\cite{CP07} processes are enriched with labels, and a testing
equivalence is defined by means of  schedulers that synchronize with
processes on the process labels. For a certain labeling, processes
$s$ and $\bar{s}$ can be equated. Although this is an elegant and
compositional solution to the problem of overestimating
probabilities in testing semantics, we believe that our approach is
more feasible in practice. In fact, the task of the schedulers and
the purpose of the process labels in \cite{CP07} in our testing
semantics have been accomplished by the rational functions formed
from the action labels.

\paragraph{Acknowledgements} We thank Jos Baeten and Erik de Vink
for their valuable comments on a draft version of this paper.

\bibliographystyle{plain}

\bibliography{testing1}{}

\end{document}